\documentclass[reqno,12pt]{amsart}

\usepackage[a4paper, 
            left=1.in,
            right=1in,
            top=1in,
            bottom=1in, 
            footskip=.25in]{geometry}

\usepackage[utf8]{inputenc}
\usepackage{microtype}
\usepackage{graphicx}
\usepackage{amssymb} 
\usepackage[colorlinks = true,
            linkcolor = blue,
            urlcolor  = blue,
            citecolor = blue,
            anchorcolor = blue]{hyperref}

\usepackage{enumerate}
\usepackage{enumitem}
\usepackage{thmtools} 
\usepackage{thm-restate}

\usepackage{microtype}

\usepackage{xcolor}

\theoremstyle{plain}
\newtheorem{thm}{}[section]
\newtheorem{lemma}[thm]{Lemma}
\newtheorem{proposition}[thm]{Proposition} 

\newtheorem{theorem}[thm]{Theorem}
\newtheorem{corollary}[thm]{Corollary}
\theoremstyle{remark}
\newtheorem{claim}[thm]{Claim}
  
 \newtheorem{remark}[thm]{Remark} 
 \newtheorem{example}[thm]{Example}

\newenvironment{claimproof}[1][\proofname]
{\proof[#1]}
{\endproof}

\usepackage{centernot}  
\usepackage{xspace}

\usepackage[scr=boondoxo]{mathalpha}  
\newcommand{\coloneqq}{\mathrel{\mathop:}\mathrel{\mkern-1.2mu}=} 

\newcommand{\struct}[1]{\mathfrak{#1}}    
\newcommand{\CSP}{\ensuremath{\mathrm{CSP}}\xspace}    
\newcommand{\age}{\ensuremath{\mathit{age}}\xspace}   
\newcommand{\fm}{\ensuremath{\mathit{fm}}\xspace}  
\newcommand{\efm}{\ensuremath{\mathit{expfm}}\xspace} 
 
\newcommand{\colours}{\ensuremath{n\text{-}\mathit{colours}}}   
\newcommand{\acts}{\curvearrowright} 
\newcommand{\hh}{\ensuremath{\mathit{ht}}\xspace}
\newcommand{\lh}{\ensuremath{\mathit{lh}}\xspace}
\newcommand{\wh}{\ensuremath{\mathit{wd}}\xspace}
\newcommand{\ar}{\ensuremath{\mathit{ar}}\xspace}   
\newcommand{\MMSNP}{\ensuremath{\mathrm{MMSNP}}\xspace}  
\newcommand{\GMSNP}{\ensuremath{\mathrm{GMSNP}}\xspace}  
\newcommand{\cplmt}[1]{\smash{\overline{#1}}} 
\newcommand{\pre}[1]{#1\smash{^{-1}}}   
\newcommand{\prexists}[1]{#1\smash{^{-1\exists}}}   
\DeclareMathOperator{\mo}{mod} 
\newcommand{\Aut}{\ensuremath{\mathrm{Aut}}\xspace}   
\newcommand{\NP}{{\textup{\textsf{NP}}}\xspace}

\newcommand{\TWONEXPTIME}{{\textup{\textsf{2NEXPTIME}}}\xspace} 
\newcommand{\TWOEXPTIME}{{\textup{\textsf{2EXPTIME}}}\xspace}  
\newcommand{\Forb}{\ensuremath{\mathrm{Forb}}\xspace}   
 
\usepackage{caption}
\AtBeginDocument{\fontsize{12pt}{15pt}\selectfont}

\begin{document}

\title[]{The Golden Path to Guarded Monotone Strict NP}

\author{Alexey Barsukov} 
\author{Michael Pinsker}
\author{Jakub Rydval}

\address{Faculty of Mathematics and Physics, Charles University, Prague, Czechia}
\email{alexey.barsukov@matfyz.cuni.cz}  

\address{Institut f\"{u}r Diskrete Mathematik und Geometrie, FG Algebra, TU Wien, Austria}
\email{$\{$michael.pinsker,jakub.rydval$\}$@tuwien.ac.at}

 \begin{abstract}   

 Guarded Monotone Strict NP (GMSNP) extends Monotone Monadic Strict NP (MMSNP) by guarded existentially quantified predicates of arbitrary arities.
 We prove that the containment and the FO-rewritability problems for GMSNP are decidable, thereby 
 settling an open question of Bienvenu, ten Cate, Lutz, and Wolter, later restated by Bourhis and Lutz.
 Our proof also comes with a \TWONEXPTIME upper bound on the complexity of the two problems, which matches the lower bounds for MMSNP due to Bourhis and Lutz. 
 
To obtain these results, we significantly improve the state of knowledge of the model-theoretic properties of GMSNP.
Bodirsky, Kn\"{a}uer, and Starke previously showed that every GMSNP sentence defines a finite union of CSPs of $\omega$-categorical structures. 
We show that these structures can be used to obtain a reduction from the containment problem for GMSNP to the much simpler problem of testing the existence of a recolouring; a careful analysis of this yields said upper bound for containment.
The upper bound for FO-rewritability is subsequently obtained by an application of several standard techniques from the theory of infinite-domain CSPs.

As our secondary contribution, we refine the construction of Bodirsky, Kn\"{a}uer, and Starke by adding a restricted form of homogeneity to the properties of these structures, making the logic amenable to future complexity classifications for query evaluation using techniques developed for infinite-domain CSPs.
\end{abstract}  

 \thanks{\emph{Michael Pinsker and Jakub Rydval}: This research was funded in whole or in part by the Austrian Science Fund (FWF) [I 5948, ESP 1571225]. For the purpose of Open Access, the authors have applied a CC BY public copyright licence to any Author Accepted Manuscript (AAM) version arising from this submission. 
 \\ \emph{Alexey Barsukov and Michael Pinsker}: This research is  funded by the European Union (ERC, POCOCOP, 101071674). Views and opinions expressed are however those of the author(s) only and do not necessarily reflect those of the European Union or the European Research Council Executive Agency. Neither the European Union nor the granting authority can be held responsible for them.}
%

\maketitle


\section{Introduction}\label{section:introduction}
\emph{Guarded Monotone Strict NP} ($\GMSNP$) is a syntactic fragment of \emph{Existential Second-Order} (ESO) logic 
describing problems of the form
\begin{center} \vspace{0.75em}
    {\it \parbox{0.8\textwidth}{Is there a colouring of the relational tuples in a given finite relational structure avoiding a fixed set of finitely many forbidden colour-patterns?}} \vspace{0.75em}
\end{center} 
It was introduced in~\cite{bienvenu2014} as a generalisation of  \emph{Monotone Monadic Strict NP without inequality} ($\MMSNP$) of Feder and Vardi~\cite{federvardi1998}, which captures similar problems concerning colourings of vertices instead of relational tuples.
Its true origins, however, go back to an earlier work of Madelaine~\cite{madelaine2009universal}, who was studying GMSNP under the names MMSNP$_2$ and FPP (standing for \emph{Forbidden Pattern Problems}).
\subsection{A brief history of GMSNP} \label{section:brief_history}

The logic class \emph{Strict NP} (SNP) consists of all problems expressible by an ESO sentence with a universal first-order part.
It was originally presented by Kolaitis and Vardi in 1987~\cite{kolaitis1987decision} as an expressive syntactic fragment of ESO which has a 0-1 law and where the associated decision problem is decidable (as opposed to  full ESO).
The study of SNP was later picked up by Feder and Vardi~\cite{federvardi1998} who showed that it has the same computational power as all of NP (it is \emph{NP-rich}).
They then presented MMSNP as a suitable candidate for a smaller yet still very expressive class which is not NP-rich (unless P=NP).
Feder and Vardi~\cite{federvardi1998} showed that every problem described by an MMSNP sentence is equivalent under polynomial-time randomised reductions to the \emph{Constraint Satisfaction Problem} (CSP) of a structure with a finite domain, i.e., the problem of testing whether a given conjunctive query is satisfiable in that structure.
They moreover conjectured that every finite-domain CSP is in \textsf{P} or \textsf{NP}-complete. 

 A standard example of a finite-domain CSP is \textsc{3-colouring}, which is the CSP of the complete graph on 3 vertices. 
Outside of the CSP framework, one would rather introduce \textsc{3-colouring} as a forbidden pattern problem, asking whether a given graph can be vertex-3-coloured while avoiding monochromatic edges (see Figure~\ref{fig:3colouring}).
To see that the two definitions coincide, one must take the (equivalent) homomorphism perspective on CSPs, where an instance is viewed as a finite structure.
Then the CSP of a structure is the problem of testing whether a given instance \emph{homomorphically} maps to that structure.
Using the homomorphism perspective, it is not hard to show that, in fact, every finite-domain CSP is definable in MMSNP. 
\begin{figure}[ht]
     \centering
 \includegraphics[width=0.5\textwidth]{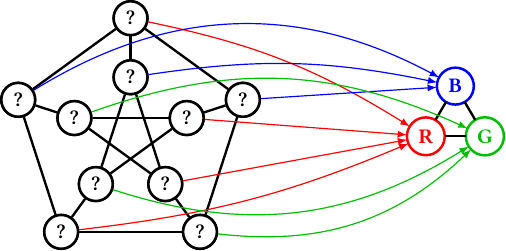} 
     \caption{A vertex-3-colouring of the Petersen graph avoiding monochromatic edges. This colouring is simultaneously a homomorphism from the Petersen graph to the complete graph on 3 vertices.}
     \label{fig:3colouring}
 \end{figure}

In 2014, Bienvenu, ten~Cate, Lutz, and Wolter~\cite{bienvenu2014} discovered a way to translate complexity classification results for certain well-behaved fragments of the logic SNP~\cite{kolaitis1987decision,papadimitriou1988optimization} to analogous statements about the complexity of \emph{Ontology-Mediated Queries} (OMQs). 
In the special case of MMSNP, they showed that there is a dichotomy between \textsf{P} and 
\textsf{NP}-completeness if and only if there is a dichotomy between \textsf{P} and 
\textsf{coNP}-completeness for the evaluation of \emph{Unions of Conjunctive Queries} (UCQs) mediated by ontologies specified in the description logic $\mathcal{ALC}$ (see~\cite{baader2003description_corrected} for more details).  
Similarly, a \textsf{P}/\textsf{coNP}-complete dichotomy for UCQs mediated by ontologies specified in \emph{Guarded First-Order} logic (GFO), another well-behaved fragment of FO, would translate to a \textsf{P}/\textsf{NP}-complete dichotomy for a strictly more expressive fragment:

\begin{center} \smallskip
    \emph{Guarded Monotone Strict NP} (GMSNP)
\end{center}  \smallskip

The Feder-Vardi conjecture was confirmed in 2017 by Bulatov and Zhuk~\cite{bulatov2017,zhuk2020} in what can be viewed as the culmination of decades of work on connections between algorithms for CSPs and universal algebra.
This implied dichotomies for MMSNP and, in turn, for the above-mentioned queries up to randomised reductions. 
Although the originally randomised reduction of MMSNP to CSPs by Feder and Vardi was derandomised by Kun in 2013~\cite{kun2013}, it still relied on a probabilistic proof of the existence of expanders of large girth.  
The first truly constructive proof of the dichotomy was obtained in 2018 by Bodirsky, Madelaine, and Mottet~\cite{bodirsky2018_article} using methods from infinite-domain constraint satisfaction, i.e., CSPs of infinite structures.  
Roughly speaking, Bodirsky, Madelaine, and Mottet~\cite{bodirsky2018_article} start with the observation from~\cite{BodDalJournal} that every MMSNP sentence defines a finite union of CSPs of (infinite) \emph{$\omega$-categorical} structures. 
 They then gradually upgrade these structures until a state is reached where they can show that their CSPs admit a uniform polynomial-time many-one reduction to a finite-domain CSP which is solvable in polynomial-time unless one of the $\omega$-categorical structures can simulate \textsc{3-colouring} via  a \emph{pp-construction}~\cite{wonderland} (in which case the sentence defines a NP-complete problem).

The question whether GMSNP also exhibits a \textsf{P}/\textsf{NP}-complete dichotomy was left open in~\cite{bienvenu2014}.
It was observed by Bodirsky, Kn\"{a}uer, and Starke~\cite{bodirsky_asnp} that similarly to MMSNP, the logic GMSNP can be studied using methods from infinite-domain constraint satisfaction.
Using this approach, the existence of a dichotomy has already been  confirmed in some specific cases, e.g., for graph orientations with forbidden tournaments~\cite{bodirsky2023forbidden,bitter2024completion, feller2024algebraic}. 
However, at the moment there is no clear road map for how exactly the question of the existence of a \textsf{P}/\textsf{NP}-complete dichotomy for GMSNP should be approached: it is apparent that the increase in the arity of the existentially quantified predicates introduces new obstacles which were not present with MMSNP. 

Hope is sparked by a recent result of Guzm\'{a}n-Pro~\cite{guzman2024gmsnp} who used the sparse incomparability lemma~\cite{kun2013}  to show that the homomorphism-sandwiching method~\cite{brakensiek2019algorithmic} for obtaining hardness reductions from finite-domain \emph{promise} CSPs fails for CSPs definable in GMSNP. 
This is in fact a good plausibility argument for the applicability of  methods from infinite-domain constraint satisfaction to GMSNP, because several NP-hardness results for finite-domain promise CSPs~(\cite[Prop.~10.1]{pcsp_bible} and~\cite[Thm.~2.2]{wrochna_zivny2020}) are inconsistent with the algebraic dichotomy conjecture for infinite-domain CSPs~\cite{barto_pinsker_journal}.

\subsection{The logic GMSNP} \label{sec:intro_one}
A GMSNP sentence  over a finite relational signature $\tau$ is an ESO sentence $\Phi$  of the form $\exists X_1,\dots, X_n \forall \bar{x}\ldotp \phi(\bar{x}),$
for a CNF-formula $\phi$ over the relational signature $\tau\cup \{X_1,\dots, X_n\}$, if it satisfies the following two conditions called the \emph{monotonicity} and the \emph{guarding} axioms, respectively: 
\begin{itemize} 
    \item    $\tau$-atoms and equalities may only appear negatively in clauses of $\phi$, i.e., the only positive atoms in clauses of $\phi$ are of the form $X_i(\bar{x})$ for some $i\in [n]$;     
    \item  for each clause $\psi$ of $\phi$ and every positive atom $X_i(\bar{x})$ in $\psi$ there exists a negative atom $\neg R(\bar{y})$ in $\psi$ with $R\in \tau \cup \{X_1,\dots, X_n\}$ such that $\bar{x}\subseteq \bar{y}$.   
\end{itemize}  
Unless stated otherwise, we will denote the set $\{X_1,\dots, X_n\}$ of the second-order variables in $\Phi$ by $\sigma$; we refer to $\sigma$ as the \emph{existential} symbols and $\tau$ as the \emph{input} symbols in $\Phi$.
We also refer to the first-order $(\tau\cup \sigma)$-sentence $\forall \bar{x}\ldotp \phi(\bar{x})$ as the \emph{first-order part} of $\Phi$.
 By a simple application of De Morgan's laws, we can rewrite $\phi$ as a formula of the form $\bigwedge_i \neg \phi_i$, where each $\phi_i$ is a conjunction of $\tau$-atoms and possibly negated $\sigma$-atoms. 
We refer to the conjuncts $\phi_i$ as the \emph{forbidden patterns} of $\Phi$.

Every $\tau$-sentence $\Phi$ defines the corresponding \emph{model-checking problem}, where the task is to decide whether a given finite $\tau$-structure $\struct{A}$ satisfies $\Phi$.
A typical example of a problem definable by a GMSNP sentence is \textsc{Edge-No-Mono-Tri}, where the task is to decide whether a given undirected simple graph, say with an edge relation denoted by $E$, can be edge-2-coloured while avoiding monochromatic triangles~\cite{garey1979computers}:   
\begin{align}\label{ex:gmsnp_introduction} 
    \exists R, B\, \forall x,y,z &\big(  \neg E(x,y) \vee   \neg  E(y,z) \vee \neg E(z,x) \vee \neg B(x,y) \vee   \neg  B(y,z) \vee \neg B(z,x)   
    \big) \nonumber \\
     {} \wedge \,  &\big(   \neg  E(x,y) \vee  \neg E(y,z) \vee \neg E(z,x) \vee  \neg R(x,y) \vee  \neg R(y,z) \vee \neg  R(z,x)   
    \big)   \\ 
   {} \wedge \,  & \big( \neg E(x,y) \vee \neg B(x,y) \vee \neg R(x,y) \big) \wedge  \big( \neg E(x,y) \vee B(x,y) \vee R(x,y) \big).\nonumber 
\end{align} 
For simplification purposes, we omit clauses requiring the symmetricity of colours on edges.
See Figure~\ref{fig:nomonotri} below for an illustration.
The forbidden patterns of this problem are the two monochromatic triangles, overlapping colours, and uncoloured edges.
\begin{figure}[ht]
     \centering
 \includegraphics[width=0.25\textwidth]{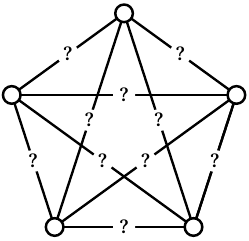}
      \qquad 
      \includegraphics[width=0.25\textwidth]{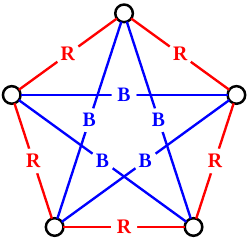}
     \caption{An edge-2-colouring of the complete graph on $5$ vertices avoiding monochromatic triangles (outer edges red, inner blue). Note that there is no vertex-2-colouring with this property.}
     \label{fig:nomonotri}
 \end{figure}

In MMSNP, we can instead formulate \textsc{Vertex-No-Mono-Tri}, where the task is to find a vertex-2-colouring instead of an edge-2-colouring (see Figure~\ref{fig:nomonotri_vertex}):
\begin{align} \label{ex:gmsnp_introduction2}
    \exists R, B\, \forall x,y,z 
   & \big(  \neg E(x,y) \vee   \neg  E(y,z) \vee \neg E(z,x) \vee \neg B(x) \vee   \neg  B(y) \vee \neg B(z)   
    \big) \nonumber  \\
   {} \wedge \, & \big(  \neg E(x,y) \vee   \neg  E(y,z) \vee \neg E(z,x) \vee \neg R(x) \vee   \neg  R(y) \vee \neg R(z)   
    \big)   \\ 
  {} \wedge \,  & \big( \neg B(x) \vee \neg R(x) \big) \wedge   \big(    B(x) \vee   R(x) \big). \nonumber
\end{align} 
\begin{figure}[ht]
     \centering
      \includegraphics[width=0.45\textwidth]{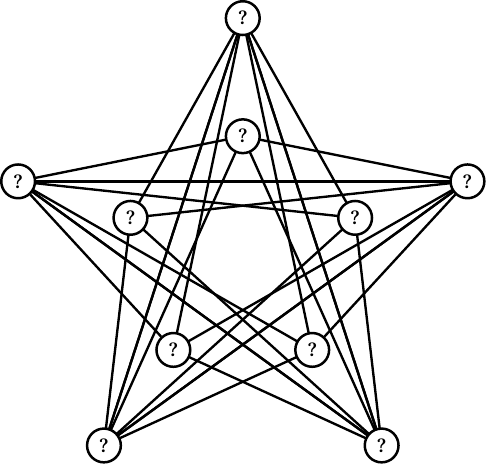}
      \qquad 
      \includegraphics[width=0.45\textwidth]{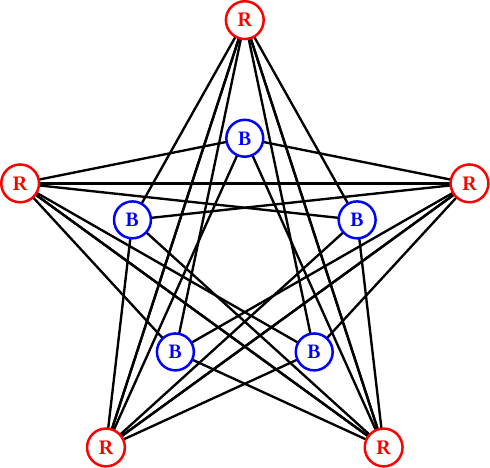}
     \caption{A vertex-2-colouring of the complement of the Petersen graph---the line graph of the complete graph on $5$ vertices---avoiding monochromatic triangles (outer vertices red, inner blue).}
     \label{fig:nomonotri_vertex}
 \end{figure}
Strictly speaking, $\GMSNP$ does \emph{not} fully contain $\MMSNP$ because the latter does not require the \emph{guarding} axiom (it instead requires the \emph{monadicity} axiom).
However,  containment does hold up to a small syntactic transformation, 
modulo which GMSNP is in fact strictly more expressive than MMSNP~\cite{bienvenu2014}.   
Namely, every MMSNP sentence $\Phi$ can be converted into a GMSNP sentence $\Gamma(\Phi)$ by expanding the base signature $\tau$ by a fresh unary ``domain'' symbol $D$ and adding negative atoms $\neg D(v)$ to each clause, for every universally quantified first-order variable $v$ appearing in that clause.
Given a $\tau$-structure $\struct{A}$, we have that $\Phi$ holds in $\struct{A}$ if and only if $\Gamma(\Phi)$ holds in $\struct{A}^+$, where $\struct{A}^+$ is the $(\tau\cup \{D\})$-structure  obtained from $\struct{A}$ by including every domain element in the relation interpreting $D$. 
On the other hand, given a $(\tau\cup \{D\})$-structure $\struct{A}$, we have that $\Gamma(\Phi)$ holds in $\struct{A}$ if and only if $\Phi$ holds in $\struct{A}^{-}$, where $\struct{A}^{-}$ is the $\tau$-structure obtained from $\struct{A}$ by forgetting $D$ and removing all domain elements which were not contained in the relation interpreting $D$. 
For example, the last two clauses in~\eqref{ex:gmsnp_introduction2} do not satisfy the guarding axiom, but we can obtain a polynomial-time equivalent  GMSNP sentence by padding~\eqref{ex:gmsnp_introduction2} with the predicate $D\in \tau$:
  \begin{align*}
     \exists R, B\,   \forall x,y,z     \big(     \neg D(x) \vee \neg D(y) \vee \neg D(z) \vee \neg E(x,y) \vee   \neg  E(y,z) \vee \neg E(z,x) \\ {} \vee \neg B(x) \vee   \neg  B(y) \vee \neg B(z)   
    \big)  \\
    \!\!\!\! {} \wedge   \big(    \neg D(x) \vee \neg D(y) \vee \neg D(z) \vee  \neg E(x,y) \vee   \neg  E(y,z) \vee \neg E(z,x)  \\ {} \vee \neg R(x) \vee   \neg  R(y) \vee \neg R(z)   
    \big)\\
  \!\!\!\! {} \wedge   \big( \neg D(x) \vee B(x) \vee R(x) \big) \wedge   \big( \neg D(x) \vee  \neg B(x) \vee \neg  R(x) \big). 
\end{align*}

\subsection{Containment \& FO-rewritability} 
Besides evaluation/model-checking,  two other computational decision problems for the aforementioned Ontology-Mediated Queries and the associated fragments of SNP were studied in~\cite{bienvenu2014}: containment and FO/Datalog-rewritability.
In the case of MMSNP, the complexity of containment and FO/Datalog-rewritability is known: all three problems are \TWONEXPTIME-complete~\cite{Collapses, mottet2021symmetries,bouhris_lutz2016}.
While rewritability questions for MMSNP seem to be intimately linked to the reduction from MMSNP to finite-domain CSPs~\cite{Collapses, mottet2021symmetries} (in particular in the case of Datalog), for the containment problem, the link seems to be much weaker~\cite{bouhris_lutz2016,bodirsky2018_article}.
This makes the study of the containment problem for GMSNP the perfect starting point for the development of tools that will be useful for the topics of rewritability and model-checking.

For reasons that will be explained later, it is important that we formulate both the containment and the FO-rewritability problems in the more general setting of \emph{SNP sentences}, which are obtained by leaving out the monotonicity and the guarding axioms. 
For an SNP sentence $\Phi$, we denote the class of all finite models of $\Phi$ by $\fm(\Phi)$.
Formally, the \emph{containment problem} can be stated as follows:

 \medskip 
 
\noindent  \underline{\textbf{Containment for SNP}}   \\
 \setlength{\tabcolsep}{0pt}    \begin{tabular}{ll}  
 \noindent INSTANCE: A pair $(\Phi_1, \Phi_2)$ of SNP sentences in a common signature $\tau$.\\
      \noindent QUESTION: Is it true that $\fm(\Phi_1) \subseteq \fm(\Phi_2)$?  
\end{tabular}  

\medskip  

The containment problem for SNP is undecidable in general; in fact, it is already undecidable for the \emph{Datalog} fragment~\cite{shmueli1993equivalence}, which limits the maximum number of non-negated atoms in clauses to one.  
The question whether containment is decidable for GMSNP was left open in~\cite{bienvenu2014,bouhris_lutz2016}; the only known result regarding the complexity of this problem is the \TWONEXPTIME lower bound of Bourhis and Lutz for the containment within MMSNP.
To see that this lower bound also applies for GMSNP, note that the map $\Gamma$ introduced in Section~\ref{sec:intro_one} is not only a polynomial-time reduction from MMSNP to GMSNP, but it also acts as a fully faithful covariant functor with respect to containment.
More specifically, for a pair $(\Phi_1, \Phi_2)$ of MMSNP $\tau$-sentences, we have 
\[ 
\fm(\Phi_1) \subseteq \fm(\Phi_2) \quad  \text{if and only if} \quad  \fm(\Gamma(\Phi_1)) \subseteq \fm(\Gamma(\Phi_2)).  
\] 
We proceed with a general formulation of the \emph{FO-rewritability} problem:
 
 \medskip 
 
\noindent  \underline{\textbf{FO-rewritability for SNP}}   \\
 \setlength{\tabcolsep}{0pt}    \begin{tabular}{ll}  
 \noindent INSTANCE: An SNP sentence $\Phi$.\\
      \noindent QUESTION: Is there a first-order sentence $\Psi$ such that $\fm(\Phi ) = \fm(\Psi)$?  
\end{tabular}  

\medskip   

 As with containment, FO-rewritability is undecidable already for the Datalog fragment of SNP~\cite{gaifman1993undecidable,hillebrand1995undecidable}.
For MMSNP, decidability of FO-rewritability was only proved quite recently~\cite{feier2019rewritability};  in contrast, the  decidability of containment for MMSNP was long known since the seminal paper of Feder and Vardi~\cite{federvardi1998}.
The decidability of FO-rewritability for GMSNP was left open in~\cite{bienvenu2014}; as with containment, a \TWONEXPTIME-lower bound is inherited from MMSNP. 

\begin{theorem}[Thms.~3 and~5 in \cite{bouhris_lutz2016}]\label{thm:lowerbound} The containment and the FO-rewritability problems for GMSNP are \TWONEXPTIME-hard.
\end{theorem}

\subsection{Primary contribution} \label{sec:contributions}

We confirm that the containment and the FO-rewritability problems for GMSNP are decidable.
Our proofs of decidability come with an upper bound on the complexity matching the lower bounds in Theorems~\ref{thm:lowerbound}.

\begin{theorem}\label{thm:2NEXPTIME_for_GMSNP} \vphantom{}
\begin{enumerate}
\item \label{item:containment} The containment problem for GMSNP is in \TWONEXPTIME.
    \item \label{item:FO_rewritability} The FO-rewritability problem for GMSNP is in \TWONEXPTIME.
\end{enumerate}
 \end{theorem}
\begin{corollary} The containment and the FO-rewritability problems for GMSNP are \TWONEXPTIME-complete.
\end{corollary}
The proof of Theorem~\ref{thm:2NEXPTIME_for_GMSNP} is based on the recolouring method, first used implicitly by Feder and Vardi~\cite[Theorem~7]{federvardi1998} in their proof of decidability of containment for MMSNP.\footnote{Recolourings were first mentioned in the work of Madelaine and Stewart~\cite{madelaine2007constraint} (see also~\cite{madelaine2010containment}).}
Roughly said, the strategy is to modify the input sentences by adding new clauses and existentially quantified symbols until a state is reached where the containment problem becomes equivalent to the much simpler problem of testing the existence of a mapping (a recolouring) between the existential symbols of the two sentences satisfying a certain condition.
In contrast to MMSNP, in order to achieve the promised \TWONEXPTIME upper bound on the complexity of containment for GMSNP using the recolouring method, we have to consider a more general framework (SNP) and also significantly generalise the notion of a recolouring. 

The first preprocessing step, common across most approaches to the containment problem for MMSNP, is a reduction to the connected case~\cite{bodirsky2018_article,bouhris_lutz2016,federvardi1998,madelaine2010containment}.
Intuitively, an MMSNP sentence is \emph{connected} if its forbidden colour-patterns are connected in the graph-theoretic sense.
This reduction can be generalised to GMSNP, as shown in~\cite{bodirsky_asnp}.  
Bodirsky, Kn\"{a}uer, and Starke~\cite[Prop.~1]{bodirsky_asnp} showed that every GMSNP sentence is logically equivalent to a finite disjunction of connected GMSNP sentences. 
 An inspection of their proof reveals that if the complexity of checking 
 containment is in $\TWONEXPTIME$ for connected GMSNP, then so it is for the entire class GMSNP.

The second preprocessing step (in the approaches to the containment for MMSNP) is typically a reduction to the biconnected case~\cite{bouhris_lutz2016,bodirsky2018_article}.
Intuitively, an MMSNP sentence is \emph{biconnected} if no forbidden colour-pattern can be 
disconnected by deleting a single vertex.
The reduction to biconnected MMSNP is the point where a significant amount of work can be outsourced to a general combinatorial result, e.g., a lemma of Erd\H{o}s in the case of~\cite{bouhris_lutz2016} or a 
theorem of Hubi\v{c}ka and Ne\v{s}et\v{r}il in the case of~\cite{bodirsky2018_article}. 
 The procedure for attaining biconnectedness from~\cite[Lem.~4.4]{bodirsky2018_article} transforms a given connected MMSNP sentence to a logically equivalent biconnected MMSNP sentence by iteratively selecting a forbidden colour-pattern which can be disconnected by deleting a single vertex and splitting it into two forbidden colour-patterns (corresponding to the two parts); the split is then marked using a fresh unary existential predicate. 
The GMSNP-analogue of this procedure does not terminate if the arities of the existentially quantified predicates are greater than $2$, and we do not see any way how this issue could be fixed.
Therefore, to progress further, we must understand how exactly the said general combinatorial results are used in the complexity analysis of the containment problem; in the present paper, we focus specifically on the approach from~\cite{bodirsky2018_article} using the mentioned theorem of Hubi\v{c}ka and Ne\v{s}et\v{r}il.

Roughly said, that  theorem implies that every connected MMSNP $\tau$-sentence $\Phi$ can be assigned a highly symmetric infinite structure $\struct{C}_{\Phi}$ in a signature extending $\tau\cup \sigma$ such that the $(\tau\cup \sigma)$-reducts of the finite substructures of $\struct{C}_{\Phi}$ are up to isomorphism precisely the finite models of the first-order part of $\Phi$; this was first observed by Bodirsky and Dalmau~\cite{BodDalJournal}.
Bodirsky, Madelaine, and Mottet~\cite{bodirsky2018_article} proved that if $\Phi$ is biconnected, then the $(\tau\cup \sigma)$-reduct of $\struct{C}_{\Phi}$ retains a certain amount of the original symmetry, and this fact is crucial in their take on the recolouring method.

Later on, Bodirsky, Kn\"{a}uer, and Starke~\cite{bodirsky_asnp} showed that also connected GMSNP sentences can be assigned said highly symmetric infinite structures; they did not comment on whether this fact is helpful in analysing the complexity of the containment problem for GMSNP. 
We show that the highly symmetric infinite structures $\struct{C}_{\Phi}$ associated with GMSNP sentences $\Phi$ can be described by SNP sentences $\Delta(\Phi)$ of size 2-exponential in the size of $\Phi$.
Subsequently, we show that the containment $\fm(\Phi_1)\subseteq \fm(\Phi_2)$ is equivalent to the existence of a particular mapping from $\struct{C}_{\Phi_1} $ to $ \struct{C}_{\Phi_2}$ which is uniquely determined by the images of substructures whose size is bounded by the maximal arity of a symbol in $\Phi_1$ or $\Phi_2$.
From this fact, we extract our more general version of recolourings for SNP.

By carefully analysing the total complexity of computing the SNP sentences $\Delta(\Phi_1)$ and $\Delta(\Phi_2)$ stemming from the theorem of Hubi\v{c}ka and Ne\v{s}et\v{r}il and of the subsequent check for the existence of a recolouring between such sentences, we conclude that the containment problem for GMSNP is in $\TWONEXPTIME$.
We remark that our methods are, in theory,  applicable in a more general setting than GMSNP; a fragment of SNP which is closely related to such applicability is  \emph{Amalgamation SNP} introduced in~\cite{bodirsky_asnp}.
We also remark that the SNP sentences $\Delta(\Phi_1)$ and $\Delta(\Phi_2)$ define a linear order over the domain of a given structure, which is at the core of our argument why the recolouring method works in our case.
It is provably not possible to define a linear order in GMSNP~\cite[Ex.~6]{bodirsky_asnp}, which indicates that this method requires us to leave the GMSNP framework. 
Finally, by applying several standard infinite-domain CSP techniques from~\cite{BodirskyM18,bodirsky2018_article}, we link the FO-rewritability of $\Phi$ to the FO-definability of a certain finite-domain CSP based on $\Delta(\Phi)$.
Since the said CSP can be parametrized by a structure of size at most 2-exponential in the size of $\Phi$, we can conclude using known results from the finite-domain CSP literature that FO-rewritability for GMSNP is in \TWONEXPTIME.

\subsection{Secondary contribution}
We analyse to what extent it is possible to prove the decidability of the containment problem for GMSNP while staying as close as possible to the original recolouring method for MMSNP.

A (vertex)-recolouring $\xi$ between two MMSNP sentences $\Phi_1$ and $\Phi_2$ is simply a mapping 
 from  
 the (unary) 
 existential  symbols of $\Phi_1$ 
  to those of  
  $\Phi_2$ whose application to the 
  relations denoted by these symbols of any model of the first-order part of $\Phi_1$ yields a model of the first-order part of $\Phi_2$. That is, the mapping $\xi$   
   does not introduce any vertex-colour patterns forbidden by $\Phi_2$ into structures which originally did not contain any patterns forbidden by $\Phi_1$  
   (see~\cite[Def.~4.25]{bodirsky2018_article}).

If $\Phi_1$ and $\Phi_2$ assert that the existential relations form a partition of all vertices (an important ingredient of the \emph{normal form} for MMSNP~\cite{bodirsky2018_article}),
then this is what one would intuitively understand under a recolouring of vertices -- hence the name. 
However, the true intention of this concept, which does not require any additional assumptions on $\Phi_1$ and $\Phi_2$ (such as the vertices being partitioned by the existential predicates) is to capture a very specific case of the containment $\fm(\Phi_1) \subseteq \fm(\Phi_2)$, which is witnessed uniformly across all structures in $\fm(\Phi_1)$.
The natural extension of vertex-recolourings to GMSNP would be \emph{edge-recolourings}, i.e., mappings 
\begin{align}
 \xi \colon R(x_1,\dots, x_n) \wedge \alpha_1(x_1,\dots, x_n)  \mapsto R(x_1,\dots, x_n) \wedge \alpha_2(x_1,\dots, x_n), \label{eq:edge_recolouring}
\end{align} 
where $R\in \tau$ and $\alpha_i$ is an atomic $\sigma_i$-formula ($i\in [2]$), 
 such that $\xi$ induces a mapping between models of the first-order parts of the two sentences.  
 The notion of a recolouring for SNP we use in the proof of Theorem~\ref{thm:2NEXPTIME_for_GMSNP} is considerably more general and  abstract.
Hence, it makes sense to ask whether one can also 
prove decidability of containment for GMSNP by simply refining
the input sentences $\Phi_1$ and $\Phi_2$ long enough, while staying in GMSNP, until the containment $\fm(\Phi_1)\subseteq \fm(\Phi_2)$ becomes equivalent to the existence of an edge-recolouring from $\Phi_1$ to $\Phi_2$.  

The central notion in analysing this possibility is \emph{recolouring-readiness}, which is comparable to the notion of a \emph{simple program} of Bourhis and Lutz~\cite{bouhris_lutz2016} or the normal form of Bodirsky, Madelaine, and Mottet~\cite{bodirsky2018_article}.
Instead of relying purely on syntactic preprocessing, in the definition of recolouring-readiness, we draw our inspiration from the proof of Theorem~\ref{thm:2NEXPTIME_for_GMSNP}.
Intuitively, a GMSNP $\tau$-sentence $\Phi$ is recolouring-ready if it can be assigned 
a highly symmetric infinite structure $\struct{C}_{\Phi}$ as in our proof of Theorem~\ref{thm:2NEXPTIME_for_GMSNP} such that the $(\tau\cup \sigma)$-reduct of $\struct{C}_{\Phi}$ retains a sufficient amount of the original symmetry for a reduction from containment to edge-recolouring to work. 
We show that every connected GMSNP sentence $\Phi$ is logically equivalent to a connected recolouring-ready GMSNP sentence $\Omega(\Phi)$.
However, the size of the smallest such $\Omega(\Phi)$ that we are able to obtain surpasses the upper bound provided in Theorem~\ref{thm:2NEXPTIME_for_GMSNP}. 
We leave it as an open question whether the precise complexity of containment for GMSNP can be determined purely from within the GMSNP framework.  

 Formally, the notion of recolouring-readiness additionally depends on a certain parameter $k\in \mathbb{N}$.
Intuitively, this parameter $k$ measures how well a given GMSNP sentence $\Phi$ can approximate 
\textsc{Digraph Acyclicity}, which is the CSP of $(\mathbb{Q};<)$, in the sense that $\Phi$ can detect  directed  cycles of size at most $k$.
It is known that no GMSNP sentence can fully express $\CSP(\mathbb{Q};<)$; see Example~6 in~\cite{bodirsky_asnp}.
\begin{restatable}{theorem}{recolouringreadiness}    \label{thm:recolouring_readiness}   
For every connected GMSNP $\tau$-sentence $\Phi$ and every $k\in \mathbb{N}$, there exists a logically equivalent connected $k$-recolouring-ready GMSNP $\tau$-sentence $\Omega(\Phi)$. 
Moreover, the function $\Omega$ is computable.
Suppose that $\Phi_1$ and $\Phi_2$ are two $k$-recolouring-ready GMSNP $\tau$-sentences with at most $k$-many variables per clause.
Then there exists an edge-recolouring from $\Phi_1$ to $\Phi_2$ if and only if $\fm(\Phi_1)\subseteq \fm(\Phi_2)$.
\end{restatable}

\subsection{Outline}
In Section~\ref{section:preliminaries}, we provide some background knowledge necessary for the presentation of our results.
Section~\ref{section:containment} contains the proof of Theorem~\ref{thm:2NEXPTIME_for_GMSNP}, where the individual steps are presented in the order given in the introduction.
Section~\ref{section:recolouring_ready} contains a proof of Theorem~\ref{thm:recolouring_readiness} and a deeper discussion of the recolouring method.

\subsection{Related work} 
The complexity of the containment problem for \emph{guarded Datalog}\textemdash the Datalog fragment of GMSNP\textemdash is known: this problem is \TWOEXPTIME-complete~\cite[Thms.~7 and~9]{bourhis:hal-01211282}. 
There is also an independent body of work on the containment and the FO-rewritability problems for (U)CQs mediated by ontologies defined by \emph{(Frontier-)Guarded Tuple Generating Dependencies}; both problems are \TWOEXPTIME-complete~\cite{containment_frontier_guarded,fo_frontier_guarded}.
It is known that such OMQs can be rewritten as guarded Datalog queries~\cite{containment_frontier_guarded}.
However, as the rewritings are very large~\cite{gottlob2014expressiveness}, the reduction of containment for guarded OMQs to containment for guarded Datalog does not yield optimal upper bounds.
Similarly, as the correspondence between GMSNP and the query language $(\mathrm{UCQ},\mathrm{GFO})$ is not efficient~\cite{bouhris_lutz2016}, our $\TWONEXPTIME$-upper bound for the complexity of the containment problem for GMSNP does not yield an optimal upper bound for the containment of UCQs mediated by GFO ontologies.

\section{Preliminaries}\label{section:preliminaries}
  
The set $\{1,\dots,n\}$ is denoted by $[n]$, and we use the bar notation $\bar{t}$ for tuples.
The \emph{component-wise action} of a function $f\colon A^n \rightarrow B$ on $k$-tuples is given by $$f\big((x_{1,1},\dots, x_{1,k}),\dots, (x_{n,1},\dots, x_{n,k}) \big)\coloneqq \big(f(x_{1,1},\dots,x_{n,1}),\dots, f(x_{1,k},\dots,x_{n,k}) \big).$$ 
We extend the containment relation on sets to tuples by ignoring the ordering on the entries. For example, we might write $X\subseteq \bar{t}$ for a set $X$ and a tuple $\bar{t}$. 

\subsection{Structures}  \label{section:prelims_structures}
A (\emph{relational}) \emph{signature} $\tau$ is a set of \emph{relation symbols}, where each $R\in\tau$ is associated with a natural number called \emph{arity}.
A (\emph{relational}) \emph{$\tau$-structure} $\struct{A}$ consists of a set $A$ (the \emph{domain}) together with the relations $R^{\struct{A}}\subseteq A^{k}$ for each $R\in \tau$ with arity $k$.
An \emph{expansion} of $\struct{A}$ is a $\sigma$-structure $ \struct{B}$ with $A=B$ such that $ \tau\subseteq \sigma$ and $R^{\struct{B}}=R^{\struct{A}}$ for each relation symbol $R\in \tau$. Conversely, we then  call $\struct{A}$ a \emph{reduct} of $\struct{B}$.
We denote the reduct of $\struct{B}$ to a subset $\tau$ of its signature by $\struct{B}^{\tau}$.
A \emph{linear-order expansion} of a $\tau$-structure $\struct{A}$ is an expansion by a single linear order denoted by $<$ (assuming  ${<}\notin \tau$); we denote such expansion by $(\struct{A},<)$.
We often do not distinguish between the symbol $<$ and the associated linear order. 
Given a class of $\tau$-structures $\mathcal{K}$, we denote by $\mathcal{K}^{<}$ the class of all $\tau\cup \{<\}$-structures which are linear-order expansions of structures from $\mathcal{K}$.
The \emph{union} of two $\tau$-structures $\struct{A}$ and $\struct{B}$ is the $\tau$-structure $\struct{A}\cup \struct{B}$ with domain $A\cup B$ and relations 
$R^{\struct{A}\cup \struct{B}}\coloneqq R^{\struct{A}}\cup R^{\struct{B}}$ for every $R\in \tau$.
If $A\cap B = \emptyset$, we call $\struct{A}\cup \struct{B}$ a \emph{disjoint union} and write $\struct{A}+\struct{B}$.
A structure is \emph{connected} if it is not the disjoint union of two structures with non-empty domains.
For a positive integer $d$ and a $\tau$-structure $\struct{A}$, the \emph{$d$-th power} of $\struct{A}$ is the $\tau$-structure $\struct{A}^d$ with the domain $A^d$ and relations as follows. For each $R\in \tau$ of arity $k$:
\begin{align*}
    R^{\struct{A}^d}=\Bigl\{\big((a_{1,1},\dots, a_{1,d}),\dots, & (a_{k,1},\dots, a_{k,d})\big)\in (A^d)^k \\ & \Big| \  (a_{1,1},\dots, a_{k,1}),\dots, (a_{1,d},\dots, a_{k,d}) \in R^{\struct{A}}\Bigr\}
\end{align*}

A \emph{homomorphism} $h\colon \struct{A} \rightarrow \struct{B}$ for $\tau$-structures $\struct{A},\struct{B}$ is a mapping $h\colon  A\rightarrow B$ that \emph{preserves} each relation of $\tau$, i.e., whenever $ \bar{t} \in R^{\struct{A}}$ for some 
 relation symbol $R\in \tau$, then $h(\bar{t})$ (computed component-wise) 
 is an element of $R^{\struct{B}}$.
We write $\struct{A} \rightarrow \struct{B}$ if $\struct{A}$ maps homomorphically into $\struct{B}$. 
The \emph{Constraint Satisfaction Problem} (CSP) of $\struct{B}$, denoted by $\CSP(\struct{B})$, is defined as the class of all finite structures which homomorphically map into $\struct{B}$.
An \emph{embedding} 
 is an injective homomorphism $h\colon \struct{A} \rightarrow \struct{B}$ that additionally satisfies the following condition: for every $k$-ary relation symbol $R\in \tau$ and $\bar{t}\in A^{k}$ we have $h(\bar{t})\in R^{\struct{B}}$ only if $\bar{t}\in R^{\struct{A}}.$
We write $\struct{A}\hookrightarrow \struct{B}$ if $\struct{A}$ embeds into $\struct{B}$. 
The \emph{age} of $\struct{B}$, denoted by $\age(\struct{B})$, is the class of all finite structures which embed into $\struct{B}$.
A \emph{substructure} of $\struct{B}$ is a structure $\struct{A}$ over $A\subseteq B$ such that the inclusion map $i\colon A\rightarrow B$ is an embedding. 
Conversely, we then call $\struct{B}$ an \emph{extension} of $\struct{A}$.
An \emph{isomorphism} is a surjective embedding.   
Two structures $\struct{A}$ and $\struct{B}$ are \emph{isomorphic} if there exists an isomorphism from $\struct{A} $ to $\struct{B}$.  
A \emph{partial isomorphism} between two structures $\struct{A}$ and $\struct{B}$ is an isomorphism from a substructure of $\struct{A}$ to a 
substructure of $\struct{B}$; if $\struct{A}=\struct{B}$, then we speak of a partial isomorphism on $\struct{A}$.

An \emph{automorphism} of $\struct{A}$ is an isomorphism from $\struct{A}$ to itself.  
The set (group) of all automorphisms of $\struct{A}$ is denoted by $\Aut(\struct{A})$.
The \emph{orbit} of a tuple $\bar{t}\in A^{k}$ in $\struct{A}$ is the set $\{g(\bar{t}) \mid g \in \Aut(\struct{A})\}.$ 
For two structures $\struct{A}$ and $\struct{B}$, a function $f\colon A \rightarrow B$ is called \emph{canonical from $\struct{A}$ to $\struct{B}$} if, for every $k\geq 1$, the component-wise action of $f$ induces a well-defined function from the orbits of $k$-tuples in $\struct{A}$ to the orbits of $k$-tuples in $\struct{B}$; that is, if any two $k$-tuples belonging to the  same orbit, say $O_1$, are mapped into the same orbit, say $O_2$; we can then write $f(O_1) = O_2$.
In other words, for every $k\geq 1$, $\bar{t}\in A^k$,  and $\alpha\in\Aut(\struct{A})$, there exists $\beta\in\Aut(\struct{B})$ such that $f(\alpha(\bar{t}))=\beta(f(\bar{t}))$.

A countable structure $\struct{A}$ is \emph{$\omega$-categorical} if, for every $k\geq 1$, there are only finitely many orbits of $k$-tuples in $\struct{A}$. 
The following lemma can be shown by a  standard compactness argument, e.g., using K\H{o}nig's tree lemma.

\begin{lemma}[Lem.~4.1.7 in~\cite{Bodirsky_book}] \label{lemma:compactness} 
Let $\struct{A}$, $\struct{B}$ be countable relational structures such that $\struct{B}$ is $\omega$-categorical.
There is a  homomorphism (embedding) from $\struct{A}$ to $\struct{B}$ if and only if for all finite substructures of $\struct{A}$ there is a  homomorphism (embedding) to $\struct{B}$.
\end{lemma} 

A $d$-ary \emph{polymorphism} $f$ of a relational $\tau$-structure $\struct{A}$ is a mapping $f\colon A^d\to A$ such that for every $R\in\tau$ and every $\bar a_1,\ldots,\bar a_d\in R^\struct{A}$, we have that $f(\bar a_1,\ldots,\bar a_d)\in R^\struct{A}$, where $f$ is applied to the tuples component-wise.
In other words, a polymorphism of $\struct{A}$ is a homomorphism from the $d$-th power $\struct{A}^d$ of $\struct{A}$ into $\struct{A}$ for some $d\in \mathbb{N}$. 
Unless $|A|\leq 1$, for every $d\in \mathbb{N}$, there always exist at least $d$ distinct polymorphisms of $\struct{A}$: the \emph{projections} $\pi_i^{d}(a_1,\dots, a_d) \coloneqq a_i$.

Let $\struct{B}$ be an expansion of $\struct{A}$.
A $d$-ary polymorphism $f$ of $\struct{A}$ is called \emph{canonical} with respect to  $\struct{B}$  
if, for every $n\in\mathbb{N}$ and $\bar a_1,\ldots,\bar a_m\in A^n$ and $\alpha_1,\ldots,\alpha_m\in \Aut(\struct{B})$, there exists $\beta\in \Aut(\struct{B})$ such that $f(\alpha_1\bar a_1,\ldots,\alpha_m\bar a_m)=\beta\circ f(\bar a_1,\ldots,\bar a_m)$.
Every $d$-ary polymorphism $f$ of $\struct{A}$ which is  canonical with respect to  $\struct{B}$ induces a well-defined $d$-ary operation $f^{\acts}$ on the set of all orbits of tuples in $\struct{B}$ (respecting their arities).

\subsection{Logic} \label{section:prelims_logic}
We assume that the reader is familiar with classical \emph{first-order} logic as well as with basic preservation properties of first-order formulas, e.g., that every first-order formula $\phi$ is preserved by isomorphisms; by embeddings if $\phi$ is existential, and by homomorphisms if $\phi$ is existential positive.
We assume that equality $=$ is always available when building first-order formulas.
We say that a first-order formula $\phi$ is \emph{$k$-ary} if it has $k$ free variables; we use the notation $\phi(\bar{x})$ to indicate that the free variables of $\phi$ are among $\bar{x}$.
Note that this does not mean that the truth value of $\phi$ depends on each entry in $\bar{x}$. 

In the present article, \emph{atomic $\tau$-formulas}, or \emph{$\tau$-atoms} for short, over a relational signature $\tau$ are of the form $R(\bar{x})$ for some $R\in \tau$ and a tuple $\bar{x}$ of first-order variables matching the arity of $R$.
For technical reasons, formulas of the form $x=y$ built using the default equality predicate are not considered atomic.
For a finite conjunction $\phi$ of $\tau$-atoms, the \emph{canonical database} of $\phi$ is the structure $\struct{A}$ whose domain consists of the variables of $\phi$ and such that $\struct{A}\models R(\bar{t})$ holds if and only if $\phi$ contains the $\tau$-atom $R(\bar{t})$ as a conjunct.
The canonical database of $\phi$ admits a homomorphism to a $\tau$-structure $\struct{B}$
if and only if $\phi$ is satisfiable in $\struct{B}$~\cite{10.1145/800105.803397}.
For a finite $\tau$-structure $\struct{A}$, the \emph{canonical query} of $\struct{A}$ 
is defined as  $ \bigwedge\nolimits_{R\in \tau}\bigwedge\nolimits_{\bar{t}\in R^{\struct{A}}} R(\bar{t})$. 
The canonical query of $\struct{A}$ is satisfiable in a $\tau$-structure $\struct{B}$ if and only if there exists a homomorphism from $\struct{A}$ to $\struct{B}$~\cite{10.1145/800105.803397}. 

Let $\struct{A}$ and $\struct{B}$ be relational structures and $d$ a natural number.
A \emph{$d$-dimensional First-Order} (FO) \emph{interpretation} of $\struct{B}$ in $\struct{A}$ is a partial surjection $\mathcal{I}\colon A^d \rightarrow B$ with the following property: the domain and the kernel of $\mathcal{I}$ have a first-order definition in  $\struct{A}$, and so does, for every relation $R$ defined by a $k$-ary atomic formula in $\struct{B}$, the $dk$-ary (preimage) relation 
\[
\mathcal{I}^{-1}(R) \coloneqq \bigl\{(a_1^1,\dots, a_d^1, \dots, a_1^k,\dots, a_d^k)\ \big|\ \big(\mathcal{I}(a_1^1,\dots, a_d^1),\dots, \mathcal{I}(a_1^k,\dots, a_d^k)\big) \in R \bigr\}\; .
\]

We denote the FO-formulas witnessing that $\mathcal{I}$ is a FO interpretation by $\phi_{\mathrm{dom}}$, $\phi_{\mathrm{ker}}$, and $\phi_R$, respectively. 

Let $\tau_1$ and $\tau_2$ be two finite relational signatures and $\mathcal{R}$ a tuple consisting of a constant $d\in\mathbb N$ and FO $\tau_1$-formulas $\phi_{\mathrm{dom}}$, $\phi_{\mathrm{ker}}$, and $\phi_R$ for every $R\in \tau_2$.
Moreover, let  $\mathcal{K}_1$ and $\mathcal{K}_2$ be two isomorphism-closed classes of finite $\tau_1$- and $\tau_2$-structures, respectively. 
Note that then we can assign to any $\tau_1$-structure $\struct{A}$ a unique $\tau_2$-structure $\mathcal{R}(\struct{A})$ with a $d$-dimensional FO interpretation in $\struct{A}$ as given by the formulas in $\mathcal R$.
We say that $\mathcal{R}$ is a \emph{FO-reduction} from $\mathcal{K}_1$ to $\mathcal{K}_2$ if, for every finite $\tau_1$-structure $\struct{A}$, we have 
$$
   \struct{A}\in \mathcal{K}_1 \quad \text{if and only if} \quad  \mathcal{R}(\struct{A}) \in \mathcal{K}_2.
$$ 
The following lemma is folklore; we include a proof for the convenience of the reader. 

\begin{lemma} \label{lemma:FO_reduction}
   If $\mathcal{K}_1$ FO-reduces to $\mathcal{K}_2$ and $\mathcal{K}_2$ is FO-definable, then so is $\mathcal{K}_1$.
\end{lemma}
\begin{proof} 
    Let $\mathcal{R}$ be a FO-reduction 
    from $\mathcal{K}_1$ to $\mathcal{K}_2$, and let $\Phi$ be a FO-sentence for $\mathcal{K}_1$. We may assume that $\Phi$ is of the form $\mathsf{Q}_1x_1 \cdots \mathsf{Q}_nx_n \ldotp \phi(x_1,\dots,x_n)$, where $\mathsf{Q}_1, \dots, \mathsf{Q}_n \in \{\exists,\forall\}$ and $\phi$ is quantifier-free. 
    Let $I_{\exists}$ and $I_{\forall}$ be the sets of $i\in [n]$ such that $\mathsf{Q}_i= \exists$ and $\mathsf{Q}_i= \forall$, respectively.
    We define $\mathcal{R}^{-1}(\Phi)$ as the FO-sentence 
    \[
     \mathsf{Q}_1 \bar{x}_1 \cdots \mathsf{Q}_n  \bar{x}_n \ldotp  \bigwedge\nolimits_{i\in I_{\exists}}  \phi_{\mathrm{dom}}(\bar{x}_i) \wedge \Big( \big(\bigwedge\nolimits_{j\in I_{\forall}} \phi_{\mathrm{dom}}(\bar{x}_j) \big) \Rightarrow \phi'(\bar{x}_1,\dots,\bar{x}_n)
    \Big),
    \]
    where $\phi'$ is obtained from $\phi$ by replacing every equality $(x_k=x_{\ell})$ with $\phi_{\mathrm{ker}}(\bar{x}_k,\bar{x}_{\ell})$ and every atomic $\tau_1$-formula $R(x_{i_1},\dots, x_{i_m})$ with $\phi_R(\bar{x}_{i_1},\dots, \bar{x}_{i_m})$.  
    Suppose that $\struct{A}$ is a finite $\tau_1$-structure.
    Then  $\struct{A} \in \mathcal{K}_1$ if and only if 
    $\mathcal{R}(\struct{A}) \in \mathcal{K}_2$ because $\mathcal{R}$ is a FO-reduction.
    Since $\mathcal{R}(\struct{A}) \in \mathcal{K}_2$ if and only if $\mathcal{R}(\struct{A}) \models \Phi$, the fact that $\struct{A} \in \mathcal{K}_1$ if and only if $\struct{A}\models  \mathcal{R}^{-1}(\Phi)$ follows directly from the definition of $\mathcal{R}^{-1}(\Phi)$. 
\end{proof}

 \subsection{Structural Ramsey theory} 
A relational structure $\struct{B}$ is \emph{homogeneous} if, for every $k\geq 1$, two tuples $\bar{t}_1,\bar{t}_2\in B^k$ lie in the same orbit if and only if the function that maps the $i$-th coordinate in $\bar{t}_1$ to the $i$-th coordinate in $\bar{t}_2$ for all $1\leq i\leq k$ is an isomorphism between two substructures of $\struct{B}$. 
In other words, a structure is homogeneous if every partial isomorphism between two finite substructures extends to an automorphism of the entire structure.

Every homogeneous structure over a finite relational signature is $\omega$-categorical, because orbits of tuples are determined by the atomic formulas satisfied by them; there are only finitely many such formulas of each arity as the signature is finite.

Homogeneous structures arise as limit objects of certain classes of finite structures. 
Let $\mathcal{K}$ be a class of finite structures in a finite relational signature $\tau$ closed under isomorphisms and substructures.  
We say that $\mathcal{K}$ has the \emph{amalgamation property} (AP) if, for all $\struct{A},\struct{B} \in \mathcal{K}$ whose substructures induced on $A\cap B$ are identical, there exist  $\struct{W}\in \mathcal{K}$ and embeddings $e\colon \struct{A}\hookrightarrow \struct{W}$, $f\colon \struct{B} \hookrightarrow \struct{W}$
such that $e|_{A\cap B} = f|_{A\cap B}.$
Note that the AP for $\mathcal{K}$ is implied by the property of being closed under unions $\struct{B}_1\cup \struct{B}_2$, also called \emph{free amalgams}; in this case, we say that $\mathcal{K}$ has the \emph{free amalgamation property}. 
\begin{theorem}[Fra\"{i}ss\'{e}, Theorem~6.1.2 in~\cite{hodges_book}] \label{theorem:fraisse_2} For a class $\mathcal{K}$ of finite structures in a finite relational signature $\tau$, the following are equivalent:
\begin{itemize}
    \item $\mathcal{K}$  is the age of a  countable homogeneous $\tau$-structure; this structure is necessarily unique up to isomorphism and called the \emph{Fra\"{i}ss\'{e}-limit} of $\mathcal{K}$;

    \item $\mathcal{K}$ is closed under isomorphisms, substructures, and has the AP.
\end{itemize}
\end{theorem} 
  
For structures $\struct{A}$ and $\struct{W}$, we denote by $\binom{\struct{W}}{\struct{A}}$ the set of all embeddings of $\struct{A}$ into $\struct{W}$.
A class $\mathcal{K}$ of structures over a common signature $\tau$ has the \emph{Ramsey property} (RP) if,  for all $\struct{A},\struct{B}\in \mathcal{K}$ and $k\in \mathbb{N}$, there exists $\struct{W}\in \mathcal{K}$ such that, for every map $f\colon \binom{\struct{W}}{\struct{A}} \rightarrow [k]$, there exists $e\in \binom{\struct{W}}{\struct{B}}$ such that $f$ is constant on the set $\bigl\{e\circ u\ \big|\  u\in \binom{\struct{B}}{\struct{A}}\bigr\} \subseteq \binom{\struct{W}}{\struct{A}}.$ 
 
Following \cite[Def.~5.3]{bodirsky2018_article}, in the present article, we call a (countable) $\omega$-categorical structure $\struct{B}$  \emph{Ramsey} if the age of the  expansion $\struct{B}^{\text{FO}}$ of $\struct{B}$ by all first-order definable relations (which is always a homogeneous structure) has the Ramsey property. 

By the theorem of Engeler, Ryll-Nardzewski and Svenonius~\cite{hodges_book}, every orbit of a tuple in a countable $\omega$-categorical structure is first-order definable.
Since therefore the first-order definable relations are precisely the unions of orbits (of the same arity), the property of being Ramsey for an $\omega$-categorical structure only depends on its automorphism group. 
Note that a homogeneous structure with a finite relational signature is Ramsey if and only if its age has the RP: this follows from the fact that every first-order definable relation is in fact definable without quantifiers.

In the proof of Theorem~\ref{thm:2NEXPTIME_for_GMSNP}, we employ structural Ramsey theory to effectively reduce the problem of containment between two infinite classes of finite structures to the existence of a certain computable function between two infinite structures. %
This reduction is achieved by first representing the two classes by infinite structures via  Theorem~\ref{theorem:fraisse_2} above, and then ensuring the computability of the function comparing the two infinite structures (if it exists) using the following statement, an immediate consequence of Theorem~5 in~\cite{bodirsky_pinsker_ramsey_canonical}. 

For a set $C$ of functions from $X$ to $Y$, denote by $\cplmt{C}$ the set of all $g \colon X \rightarrow Y$ such that, for every finite $S\subseteq X$, there exists $f\in C$ satisfying $f|_S=g|_S$.~\footnote{This notation stands for the \emph{closure} of $C$ with respect to the topology of pointwise convergence.}

\begin{theorem}[cf.~\cite{bodirsky_pinsker_tsankov}, Thm.~5 in~\cite{bodirsky_pinsker_ramsey_canonical}]   \label{th:canonical_ramsey} 
 Let $\struct{A}$, $\struct{B}$ be countable $\omega$-categorical   relational structures such that $\struct{A}$ is Ramsey, and let $\struct{A}',\struct{B}'$ be reducts of these structures.  If there exists an embedding $f'$ from $\struct{A}'$ to $\struct{B}'$, then there also exists an embedding 
 \begin{align}
f\in \overline{\{ \beta  \circ  f' \circ  \alpha  \mid  \alpha  \in  \Aut(\struct{A}), \beta  \in  \Aut(\struct{B})\}} \label{eq:closure}
\end{align}
from $\struct{A}'$ to $\struct{B}'$ that is canonical when viewed as a mapping from $\struct{A}$  to $\struct{B}$.   
\end{theorem}

We briefly explain how to obtain Theorem~\ref{th:canonical_ramsey} from Theorem~5 in~\cite{bodirsky_pinsker_ramsey_canonical}.
\begin{proof}[Derivation of Theorem~\ref{th:canonical_ramsey} from Theorem~5 in~\cite{bodirsky_pinsker_ramsey_canonical}]
Let $f'\colon  \struct{A}' \rightarrow \struct{B}'$ be an embedding. 
In~\cite[Thm.~5]{bodirsky_pinsker_ramsey_canonical}, we select $\mathbf{G}\coloneqq \Aut(\struct{A})$ and $\mathbf{H}\coloneqq \Aut(\struct{B})$.
By definition, the age of the (homogeneous) expansion $\struct{A}^{\text{FO}}$ of $\struct{A}$ by all first-order definable relations has the Ramsey property, and $\Aut(\struct{A}^{\text{FO}})=\Aut(\struct{A})$. By the results of~\cite{kechris2005fraisse},  $\mathbf{G}$ is \emph{extremely amenable}; since~\cite{kechris2005fraisse} uses slightly different definitions, we refer to~\cite[Theorem 5.8]{hubickakonecny} for this precise statement. 
It follows that the prerequisites of Theorem~5 in~\cite{bodirsky_pinsker_ramsey_canonical} are satisfied, and hence there exists a mapping $f$ satisfying~\eqref{eq:closure} that is canonical from $\struct{A}$ to $\struct{B}$.
From~\eqref{eq:closure} it is easy to conclude that $f$ is an embedding from $\struct{A}'$ to $\struct{B}'$.
\end{proof}

In the proof of Theorem~\ref{thm:2NEXPTIME_for_GMSNP}\eqref{item:FO_rewritability}, we will also need the following corollary to Theorem~5 in~\cite{bodirsky_pinsker_ramsey_canonical} for polymorphisms of reducts of countable homogeneous Ramsey structures with finite relational signatures. 
The definition of $\cplmt{C}$ given above Theorem~\ref{th:canonical_ramsey} naturally extends to sets $C$ of functions from $X^m$ to $Y$ for a fixed $m\in \mathbb{N}$.
\begin{theorem}[cf.~Cor.~6 in~\cite{bodirsky_pinsker_ramsey_canonical}]    \label{th:canonical_ramsey_higher_ary} 
 Let $\struct{B}$ be a countable homogeneous Ramsey structure with a finite relational signature and let $\struct{A}$ be a reduct of $\struct{B}$.  
 Then, for every $m\in \mathbb{N}$ and every $m$-ary polymorphism $f'$ of $\struct{A}$, there exists an $m$-ary polymorphism 
 \begin{align*}
f\in \overline{\{ \beta  \circ  f' \circ  (\alpha_1,\dots, \alpha_m)  \mid  \alpha_1,\dots, \alpha_m,\beta  \in  \Aut(\struct{B})\}}  
\end{align*}
of $\struct{A}$ that is canonical with respect to $\struct{B}$.
\end{theorem}

\section{Containment \& FO-rewritability} 
\label{section:containment} 

To accurately measure the size of SNP sentences $\Phi$, we introduce the following four parameters:
\begin{itemize}
    \item the \emph{height} $\hh(\Phi)$ denotes the number of relation symbols in $\Phi$;
    \item the \emph{length} $\lh(\Phi)$ denotes the number of clauses in $\Phi$; 
    \item the \emph{width} $\wh(\Phi)$ denotes the maximum number of variables per clause in $\Phi$;
    \item the \emph{arity} $\ar(\Phi)$ denotes the maximum arity among all relation symbols in $\Phi$.  
\end{itemize}

\subsection{Connectedness}  \label{section:connectedness}

 An SNP $\tau$-sentence $\exists X_1,\dots, X_n \forall \bar{x}\ldotp \phi(\bar{x})$ is \emph{connected} in the sense of~\cite{bodirsky_asnp,Bodirsky_book} if, for each clause $\psi$ in $\phi$, the following $(\tau\cup \sigma)$-structure $\struct{N}_{\psi}$ is connected: the domain consists of all variables in $\psi$, and $\bar{t}\in R^{\struct{N}_{\psi}}$ if and only if $\neg R(\bar{t})$ is a disjunct in $\psi$ ($R\in \tau\cup \sigma$). 
For example, the sentences in equations~\eqref{ex:gmsnp_introduction} and~\eqref{ex:gmsnp_introduction2} from the introduction are connected, while the sentence
\begin{align}
\exists B,R\, \forall x,y,u,v\;  \big( \neg E(x,y) \vee \neg  E(u,v) \vee \neg  R(x,y) \vee B(u,v)\big) \label{eq:clause}
\end{align}
is not because 
the variable sets $\{x,y\}$ and $\{u,v\}$ partition the domain of the structure $\struct{N}_{\psi}$ associated to the unique clause in eq.~\eqref{eq:clause} into two disjoint substructures.
If $\Phi$ is a connected SNP sentence, then $\fm(\Phi)$ is preserved by taking disjoint unions, i.e., if $\struct{A},\struct{B} \in \fm(\Phi)$ have disjoint domains, then $\struct{A}+\struct{B} \in \fm(\Phi)$ ~\cite[Prop.~1.4.11]{Bodirsky_book}.

The next proposition was essentially proved in~\cite[Prop.~1]{bodirsky_asnp}, except that the authors did not comment on the complexity of the procedure arising from their proof.
\begin{proposition}[Proposition~1 in~\cite{bodirsky_asnp}] \label{prop:connected}
Every GMSNP sentence $\Phi$ is logically equivalent to a disjunction of connected $\GMSNP$ sentences $\Phi_1\vee\dots\vee\Phi_\ell$ over the same input  and existential signatures  with the following properties: 
\begin{itemize}
    \item  for $i\in [\ell]$, $\lh(\Phi_i)$ and $\wh(\Phi_i)$ are polynomial in $\lh(\Phi)$ and $\wh(\Phi)$, respectively;
    \item $\ell$ is 1-exponential in $\lh(\Phi)\cdot\wh(\Phi)$;
    \item the disjunction can be computed from $\Phi$ in deterministic 1-exponential time. 
\end{itemize} 
\end{proposition}

Using Proposition~\ref{prop:connected} we can show that when studying the complexity of the the containment and the FO-rewritability problems for GMSNP, we may without loss of generality restrict ourselves to connected GMSNP.   

\begin{proposition} \label{cor:connected_matters_not}
    If the containment problem for connected GMSNP sentences is in $\TWONEXPTIME$, then so is the containment problem for general GMSNP sentences.
   If additionally the FO-rewritability problem for connected GMSNP sentences is in $\TWONEXPTIME$, then so is the FO-rewritability problem for general GMSNP sentences. 
\end{proposition}
\begin{proof} To prove the first statement, we start with an auxiliary claim.
\begin{claim}\label{claim:connected}
 Suppose that $\Phi\coloneqq\Phi_1\vee\dots\vee\Phi_k$ and $\Psi\coloneqq\Psi_1\vee\dots\vee\Psi_\ell$ are two finite disjunctions of connected $\GMSNP$ sentences. 
 Then the following are equivalent:
 \begin{enumerate}
     \item \label{item:disj1} $\fm(\Phi)\subseteq \fm(\Psi)$;
     \item \label{item:disj2}  for every $i\in[k]$, there exists $j\in[\ell]$ such that $\fm(\Phi_i)\subseteq \fm(\Psi_j)$. 
 \end{enumerate} 
\end{claim}   
\begin{claimproof}
``(\ref{item:disj1})$\Rightarrow$(\ref{item:disj2})'' Suppose, on the contrary, that for some $i\in[k]$, there is no $j \in [\ell]$ as claimed, i.e.,~for every $j\in[\ell]$ there exists a structure $\struct{A}_j\in\fm(\Phi_i)$ such that $\struct{A}_j\not\in\fm(\Psi_j)$.
Let $\struct{A}$ be the structure obtained by taking the disjoint union of isomorphic copies of $\struct{A}_1, \dots,\struct{A}_\ell$.
Then, the connectedness of $\Phi_i$ implies that $\struct{A}\in\fm(\Phi_i)$ and consequently $\struct{A}\in\fm(\Phi)$.
As $\fm(\Phi)\subseteq\fm(\Psi)$, we have $\struct{A}\in\fm(\Psi)$.
So, for some $j\in[\ell]$, we have $\struct{A}\in\fm(\Psi_j)$.
This means that there is a $\sigma$-expansion $\struct{A}^\sigma$ of $\struct{A}$ such that $\struct{A}^\sigma\models\forall\bar{x}\ldotp \psi_j(\bar{x})$, where $\sigma$ is the existential signature of $\Psi_j$ and $\psi_j(\bar{x})$ is the quantifier-free part of $\Psi_j$.
As $\forall\bar{x}\ldotp \psi_j(\bar{x})$ is a universal first-order sentence, we also have $\struct{A}_j^\sigma\models\forall\bar{x}\ldotp \psi_j(\bar{x})$, where $\struct{A}_j^\sigma$ is the substructure of $\struct{A}^\sigma$ induced on the elements of $\struct{A}_j$.
This contradicts the assumption that $\struct{A}_j\not\in\fm(\Psi_j)$. Hence, the right-hand side of the equivalence follows.

``(\ref{item:disj2})$\Rightarrow$(\ref{item:disj1})'' This direction is trivial.
\end{claimproof}
 
 By Proposition~\ref{prop:connected}, we have $\Phi = \Phi_1\vee\dots\vee\Phi_{k}$ and $\Psi = \Psi_1\vee\dots\vee\Psi_{\ell}$ such that all of the  $\Phi_i$ are connected and $\lh(\Phi_i)=p(\lh(\Phi_i))$, $\wh(\Phi_i)=p(\wh(\Phi))$, and $k=2^{p(|\Phi|)}$, where $p(x)$ is some polynomial independent of $\Phi$; and similarly for all of the $\Psi_j$.
    By Claim~\ref{claim:connected}, to check that $\Phi\subseteq \Psi$, one goes through all the mappings $f\in [\ell]^{[k]}$ and, for every such $f$ and every $i\in[k]$, checks whether $\Phi_i\subseteq\Psi_{f(i)}$.
    Let $T(\Phi_i,\Psi_j)$ be an upper bound on the time needed to test $\Phi_i\subseteq \Psi_j$.
    Then, the total time needed to test $\Phi\subseteq \Psi$ is at most $\ell^{k}\cdot k\cdot \max_{i,j} T\bigl(\Phi_i,\Psi_j\bigr)$, from which the statement follows.

    For the second statement in the proposition, we will also need an auxiliary claim.
 
\begin{claim}[cf.~Proposition~3.3 in~\cite{bodirsky2018_article}] \label{claim:connected2}
 Suppose that $\Phi\coloneqq\Phi_1\vee\dots\vee\Phi_k$ is a disjunction of connected $\GMSNP$ sentences for which there is no subset $I\subsetneq [k]$ such that $\Phi$ is  logically equivalent to $\bigvee_{i\in I} \Phi_i$.
 Then the following are equivalent:
 \begin{enumerate}
     \item \label{item:disj1a} There exists a FO-sentence $\Psi$ such that $\fm(\Phi)=\fm(\Psi)$.
     \item \label{item:disj2a}  For every $i\in[k]$, there exists a FO-sentence $\Psi_i$ such that $\fm(\Phi_i)=\fm(\Psi_i)$.
     \end{enumerate}  
\end{claim}
\begin{claimproof}
``(\ref{item:disj2a})$\Rightarrow$(\ref{item:disj1a})'' This direction is trivial.
 
 ``(\ref{item:disj1a})$\Rightarrow$(\ref{item:disj2a})'' Since $\Phi$ is monotone, the class $\overline{\fm(\Phi)}$ of all finite $\tau$-structures which do not satisfy $\Phi$ is closed under taking homomorphic images.
 Since $\overline{\fm(\Phi)}$ is definable in FO by $\neg \Psi$, it follows from the homomorphism preservation theorem for finite structures~\cite[Theorem~1.7]{rossman2008homomorphism} that 
it is defined by an existential positive sentence.
We may assume that this sentence is of the form $(\exists \bar{x}_1  \ldotp   \psi_1)  \vee \cdots \vee (\exists \bar{x}_n \ldotp   \psi_n)$, where each $\psi_i$ is a conjunction of $\tau$-atoms. Hence, $\fm(\Phi)$ is defined by $(\forall \bar{x}_1  \ldotp \neg \psi_1)  \wedge \cdots \wedge (\forall \bar{x}_n \ldotp \neg \psi_n)$, and we may assume that this sentence is the sentence $\Psi$.

Fix an arbitrary  $i\in [k]$. Since  there is no subset $I\subsetneq [k]$ such that $\Phi$ is  logically equivalent to $\bigvee_{i\in I} \Phi_i$, there exists $\struct{A}_i \in\fm(\Phi_i)$ such that $\struct{A}_i\not\in \fm(\Phi_j)$ for every $j\in [k]\setminus \{i\}$.
Let us fix such a structure $\struct{A}_i$. 
For a given (non-empty) finite $\tau$-structure $\struct{A}$, we assume without loss of generality that $\struct{A}$ and $\struct{A}_i$ have disjoint domains; otherwise we replace $\struct{A}_i$ with an isomorphic copy of itself. 
We claim that
\begin{align}
\struct{A}\models \Phi_i \quad \text{if and only if} \quad \struct{A}_i+ \struct{A} \models \Phi. \label{eq:weird_corespondence}
\end{align} 
Indeed, if $\struct{A}\models \Phi_i$, then, by the connectedness of $\Phi_i$, we have that $\struct{A}_i+\struct{A} \models \Phi_i$ and hence $\struct{A}_i+\struct{A} \models \Phi$.
Conversely, if $\struct{A}_i+\struct{A} \models \Phi$, then there exists $j\in [k]$ such that 
$\struct{A}_i+\struct{A} \models \Phi_j$.
Since $\Phi_j$ is monotone, $\fm(\Phi_j)$ is preserved under taking substructures.
Hence, it must be the case that $\struct{A} \models \Phi_j$ and $\struct{A}_i \models \Phi_j$.
Since $\struct{A}_i \models \Phi_j$, it follows from our assumption on $\struct{A}_i$ that $j=i$. Hence, $\struct{A} \models \Phi_i$.

Recall that \eqref{eq:weird_corespondence} holds if and only if $\struct{A}_i+ \struct{A} \models \Psi.$ 
We now construct, independently of $\struct{A}$, a FO-sentence $\Psi_i$ with the property that $\struct{A}_i+ \struct{A} \models \Psi$ if and only if $\struct{A} \models \Psi_i$; since then $\Phi_i$ is equivalent to $\Psi_i$, this finishes our proof.

To do so, let for every $j\in [n]$ the set $S_j$ consist of all pairs $(\theta,\theta')$ of subconjunctions of the conjunction $\psi_j$ which partition $\psi_j$ into two parts with disjoint sets of variables, and such that $\struct{A}_i\models\exists \bar{x}_j. \theta'$. 
 We claim that 
$$
\struct{A}\models  \bigwedge\nolimits_{(\theta,\theta')\in S_{j}} \forall \bar{x}_{j}\ldotp \neg  \theta
\quad 
\text{if and only if} \quad 
\struct{A}_i + \struct{A}\models  \forall \bar{x}_j \ldotp \neg \psi_j.
$$
Indeed, assume first that the left-hand side does not hold, i.e.,  there exists $(\theta,\theta')\in S_{j}$ such that $\struct{A} \models \exists \bar{x}_{j}\ldotp \theta$.
Then this together with  $\struct{A}_i \models \exists \bar{x}_{j} \ldotp  \theta'$ implies 
$\struct{A}_i+\struct{A} \models \exists \bar{x}_j \ldotp \psi_j$ because the witnesses for the two existential statements can be combined since $\theta,\theta'$ have disjoint sets of variables. For the other direction, suppose that the right-hand side does not hold, i.e.~$\struct{A}_i + \struct{A} \models \exists \bar{x}_j \ldotp \psi_j$, as witnessed by a tuple $\bar a$.
Recall that each $\psi_j$ is a conjunction of $\tau$-atoms.
Since
every tuple in a $\tau$-relation of $\struct{A}_i + \struct{A}$ is contained either in $A_i$ or in $A$, we can partition $\psi_j$ into two subconjunctions $\theta,\theta'$, where $\theta$ contains precisely those conjuncts which are witnessed by $\bar a$ within $\struct{A}$, and $\theta'$ contains those witnessed within $\struct{A}_i$. But then $\struct{A}\models\exists \bar x_j\theta$.

To finish the proof, we set
$$\Psi_i\coloneqq \bigwedge\nolimits_{j\in [n]}\bigwedge\nolimits_{(\theta,\theta')\in S_{j}} \forall \bar{x}_{j}\ldotp \neg  \theta.
$$  
Then we have $\struct{A}\models  \Psi_i$ if and only if  $\struct{A}_i + \struct{A}\models  \Psi $, as desired.
\end{claimproof}

 \medskip By Proposition~\ref{prop:connected}, we have that $\Phi = \Phi_1\vee\dots\vee\Phi_{k}$ for some connected GMSNP sentences $\Phi_i$ such that $\lh(\Phi_i)=p(\lh(\Phi_i))$, $\wh(\Phi_i)=p(\wh(\Phi))$, and $k=2^{p(|\Phi|)}$, where $p(x)$ is some polynomial independent of $\Phi$.
Since the containment problem for connected GMSNP is in $\TWONEXPTIME$, we can ensure within this time complexity that there is no subset $I\subsetneq [k]$ such that $\Phi$ is  logically equivalent to $\bigvee_{i\in I} \Phi_i$. 
Then, by Claim~\ref{claim:connected2}, to check FO-rewritability for $\Phi$, we can simply check this for each individual $\Phi_i$; let $T'(\Phi_i)$ be an upper bound on the runtime.
Then, the total time needed to test FO-rewritability for $\Phi$ is at most $$k^{k}\cdot k\cdot \max\nolimits_{i,j} T\bigl(\Phi_i,\Phi_j\bigr) + k\cdot \max\nolimits_{i} T'\bigl(\Phi_i\bigr),$$ from which the statement follows.
\end{proof}

\subsection{Recolourings} \label{section:recolourings}

We call a relational structure \emph{standard} if its domain is $[n]$ for some $n\in \mathbb{N}$.
 For an SNP $\tau$-sentence $\Phi$, we denote the class of all finite models of the first-order part of $\Phi$  (over the signature $\tau\cup \sigma$) by $\efm(\Phi)$.    
We define the $n$-\emph{colours} of $\Phi$, denoted $\colours(\Phi)$, as the set of all standard structures from $\efm(\Phi)$ of size $\leq n$.
 A \emph{recolouring} between two SNP $\tau$-sentences $\Phi_1$ and $\Phi_2$ is a mapping $\xi$ from $\colours(\Phi_1)$ to $\colours(\Phi_2)$, where $n\coloneqq \max(\ar(\Phi_1),\ar(\Phi_2))$, with the two following properties.
 First, the $\tau$-reducts of every $n$-colour of $\Phi_1$ and of its $\xi$-image are identical.
 Secondly, for every $\struct A\in \efm(\Phi_1)$, there is a structure $\xi'(\struct A)\in \efm(\Phi_2)$ on the same domain such that for all $\struct{T} \in \colours(\Phi_1)$ and every embedding $e\in \binom{\struct{A}}{\struct{T}}$ we have $e\in \smash{\binom{\xi'(\struct{A})}{\xi(\struct{T})}}$. 
 In other words, the following  extension $\xi'$ of $\xi$ is a well-defined mapping from $\efm(\Phi_1)$ to $\efm(\Phi_2)$: 
\begin{center} \vspace{0.75em}
\framebox{\parbox{0.73\textwidth}{ 
For every $\struct{A}\in \efm(\Phi_1)$, the structure $\xi'(\struct{A})$ on the same domain    as $\struct{A}$ is obtained by replacing for every  $\struct{T} \in \colours(\Phi_1)$ and   every embedding $e\in \binom{\struct{A}}{\struct{T}}$ the substructure $e(\struct{T})$ of $\struct A$ by $e(\xi(\struct{T}))$.
}} \vspace{0.75em}
\end{center}

Note that the second formulation makes it clear that $\xi'(\struct A)$ is unique, as $n$ is equal to $\max(\ar(\Phi_1),\ar(\Phi_2))$.
For the same reason, we have that $\struct{A}^\tau = \xi'(\struct A)^\tau$.  
Also observe that in the second  formulation the well-definedness of $\xi'$ is a non-trivial property since the required replacements might be on overlapping substructures. 
We remark that this definition of a recolouring is compatible with the definition of a recolouring for the logic MMSNP given in~\cite{madelaine2010containment,bodirsky2018_article}.
There, the unary existential symbols were interpreted as colours, and a recolouring was simply a mapping between the existential symbols which does not introduce any forbidden patterns.
It is noteworthy, however, that in that case the compatibility on overlapping structures is trivially satisfied.

 \begin{figure}[ht]
     \centering
      \includegraphics[width=0.98\textwidth]{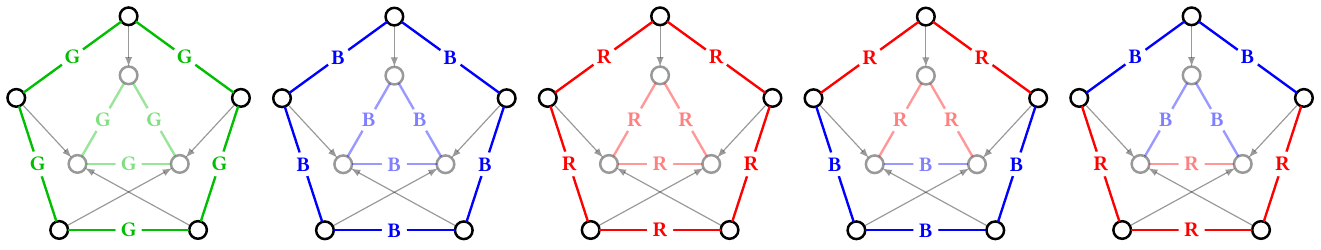}
     \caption{Pentagon-shaped forbidden colour-patterns of $\Phi_1$ from Example~\ref{ex:recolouring} and their triangle-shaped homomorphic images (which are \emph{implicitly} forbidden).}
     \label{fig:colour_patterns}
 \end{figure} 
 
\begin{example} \label{ex:recolouring}
Consider the GMSNP $\{E\}$-sentence ($E$ is binary) 
$$
\Phi_1\coloneqq \exists  R,G,B\,  \forall x_1,\dots,x_5\ldotp \phi_1(x_1,\dots,x_5)
$$
whose forbidden colour-patterns are the edge-colourings of pentagons as in  Figure~\ref{fig:colour_patterns};  
we additionally forbid uncoloured edges, asymmetric colours, and overlapping colours in order to ensure that the colours partition the $E$-edges:
\begin{align}
 & \big(\neg E(x_1,x_2) \vee \neg E(x_2,x_1) \vee R(x_1,x_2) \vee G(x_1,x_2) \vee B(x_1,x_2) \big) \label{eq:undirected}\\   
 {} \wedge {} &  \big(  \neg R(x_1,x_2) \vee \neg G(x_1,x_2)  \big) \wedge \big(  \neg R(x_1,x_2) \vee \neg G(x_2,x_1)  \big) \nonumber
 \\
 {}    \wedge {} & \big( \neg G(x_1,x_2) \vee \neg B(x_1,x_2)  \big) \wedge \big( \neg G(x_1,x_2) \vee \neg B(x_2,x_1)  \big)    \nonumber\\
 {} \wedge {} & \big(  \neg B(x_1,x_2) \vee \neg R(x_1,x_2)  \big)  \wedge  \big(  \neg B(x_1,x_2) \vee \neg R(x_2,x_1)  \big) \nonumber
\end{align}
 Note that, according to eq.~\eqref{eq:undirected}, $\Phi_1$ only requires symmetric $E$-edges of the form $E(x_1,x_2)\wedge E(x_2,x_1)$ to be coloured. 
Similarly, we define the GMSNP $\{E\}$-sentence
$$
\Phi_2\coloneqq \exists  P,G\, \forall x_1,x_2,x_3\ldotp \phi_2(x_1,x_2,x_3)   
$$
so that it forbids purple and green triangles.  
Now consider a map $\xi$ that sends undirected green coloured $E$-edges  edges to themselves, 
while merging the red and the blue colours to purple (on undirected $E$-edges).
The remaining structures in $\colours(\Phi_1)$, where $n=2$, consist either of a single vertex or two vertices that are not-$E$-related in at least one direction, possibly related with $\{R,G,B\}$-relations. 
They are mapped to $\colours(\Phi_2)$ in an arbitrary way consistent with the first recolouring condition, i.e., that the $\{E\}$-reducts are identical.
To show that $\xi$ is a recolouring, we must show that it does not produce green or purple triangles in edge-coloured graphs that originally did not admit a homomorphism from any of the five forbidden colourings of pentagons.
Clearly, $\xi$ cannot produce a green triangle:  the only way to produce a green triangle would be  to start with one, and a green triangle admits a homomorphism from a green pentagon, which is forbidden.
We can argue similarly in the case of a purple triangle, which can only be obtained with $\xi$ from one of the four red/blue-coloured triangles in Figure~\ref{fig:colour_patterns}.
All four of them admit a homomorphism from one of the forbidden coloured pentagons, as depicted in the figure. Hence, $\xi$ is a recolouring. 
\end{example}
\begin{lemma}\label{lemma:recolouring_nexptime} 
The existence of a recolouring between two SNP sentences $\Phi_1$ and $\Phi_2$ can be tested non-deterministically in time $$\mathcal{O}\Bigl(\lh(\Phi_1)\cdot \lh(\Phi_2) \cdot 2^{\wh(\Phi_1)\cdot \wh(\Phi_2)\cdot \hh(\Phi_1)\cdot\hh(\Phi_2) \cdot  2^{\ar(\Phi_1)\cdot \ar(\Phi_2)}}\Bigr).$$   
\end{lemma}%

\begin{proof} Let $n\coloneqq \max(\ar(\Phi_1),\ar(\Phi_2))$.
    Note that $|\colours(\Phi_1)|$ is bounded by the number of $(\tau\cup\sigma_1)$-structures on at most $n$ elements.
    Therefore, it is at most $N\coloneqq2^{\hh(\Phi_1)\cdot n^{\ar(\Phi_1)}}$.
    Let $\xi$ be any mapping from $\colours(\Phi_1)$ to $\colours(\Phi_2)$.
    In order to verify that $\xi$ is a recolouring, it suffices to check the following three conditions.
    \begin{enumerate}[label={\roman*.}] 
    \item \label{item:rec1} For every $\struct{T}\in\colours(\Phi_1)$, the $\tau$-reducts of $\struct{T}$ and $\xi(\struct{T})$ are identical.
    \item \label{item:rec2} For every pair $\struct{T},\struct{T}'\in\colours(\Phi_1)$, we have that every partial isomorphism $i$ from $\struct{T}$ to $\struct{T}'$ is also a partial isomorphism from $\xi(\struct{T})$ to $\xi(\struct{T}')$.
    This condition is necessary for the existence of a well-defined extension $\xi'$.  
    \item \label{item:rec3} The extension $\xi'$ maps $\efm(\Phi_1)$ to $\efm(\Phi_2)$.
\end{enumerate}

One can check Condition~\ref{item:rec1} by considering the $\xi'$-image of every element of $\colours(\Phi_1)$.
Since $|\colours(\Phi_1)|\leq N$, it can be checked in time $\mathcal{O}\bigl(\hh(\Phi_1)n^nN\bigr)$.
Suppose now  that condition~\ref{item:rec1} is satisfied.

To check condition~\ref{item:rec2}, one goes through all pairs of elements of $\colours(\Phi_1)$ and through all their substructures and checks whether the condition holds. 
This can be done in time $\mathcal{O}\bigl(N^2 (2n)^{2n}\hh(\Phi_2)\bigr)$. 
Suppose now that conditions~\ref{item:rec1},~\ref{item:rec2} are satisfied.

If condition~\ref{item:rec3} is not satisfied, then $\xi'$ is well-defined but introduces a forbidden pattern of $\Phi_2$ into a structure that originally omitted all forbidden patterns of $\Phi_1$.
More specifically, there exists a $(\tau\cup \sigma_1)$-structure $\struct{D}$ of size at most $\wh(\Phi_2)$ such that $\xi'(\struct{D})\notin \efm(\Phi_2)$ while $\struct{D} \in \efm(\Phi_1)$.
So, to check this condition, one goes through all standard $(\tau\cup \sigma_1)$-structures of size at most $\wh(\Phi_2)$; there are at most $2^{\wh(\Phi_2)^n\hh(\Phi_1)}$ of them.
Then, for each such structure $\struct{D}$, one checks in time 
$$\mathcal{O}\bigl(\lh(\Phi_1)\wh(\Phi_2)^{\wh(\Phi_1)}\hh(\Phi_1)n^n+\lh(\Phi_2)\wh(\Phi_2)^{\wh(\Phi_2)}\hh(\Phi_2)n^n\bigr)$$
the satisfiability of the first-order parts of $\Phi_1$ and $\Phi_2$ in $\struct{D}$ and in $\xi'(\struct{D})$, respectively.
In sum, the total time needed to check condition~\ref{item:rec3} is in $$\mathcal{O}\bigl(2^{\wh(\Phi_2)^n\hh(\Phi_1)}\bigl(\lh(\Phi_1)+\lh(\Phi_2)\bigr)\wh(\Phi_2)^{\wh(\Phi_1)+\wh(\Phi_2)}n^n(\hh(\Phi_1)+\hh(\Phi_2))\bigr).$$
This finishes the proof.
\end{proof}

\subsection{The theorem of Hubi\v{c}ka and Ne\v{s}et\v{r}il}

A general construction method of homogeneous Ramsey structures in a finite signature was provided by Hubi\v{c}ka and Ne\v{s}et\v{r}il~\cite{hubickanesetril2019}.
We use a variant of their result, which can be found in the appendix of~\cite{bodirsky2018_article} (Theorem~A.4 in the print version, or Theorem~A.3 in the arXiv version).

The following notion, originating from~\cite{hubicka2015}, is of essential importance in the present article.
A \emph{piece} of a structure $\struct{B}$ is a pair $(\struct{P},\bar{t})$, where $\struct{P}$ is a proper substructure of $\struct{B}$ (with domain $P  \subsetneq B$)
and $\bar t$ (called the \emph{root} of $(\struct{P},\bar{t})$) is a tuple with pairwise distinct entries enumerating a non-empty subset $T$ of $P$  such that every tuple $\bar s$ contained in a relation of $\struct B$ is completely contained either in $P$ or in $T\cup (B\setminus P)$.  

 To every piece $(\struct{P},\bar{t})$, we assign a relational symbol of arity equal to the length of $\bar t$, which for the sake of brevity will be denoted by the piece itself.   
For an atomic formula $(\struct{P},\bar{t})(\bar{x})$ using such a symbol, let $\pre{(\struct{P},\bar{t})(\bar{x})}$ be the $\tau$-formula obtained from the canonical query of $\struct{P}$ by substituting $\bar{x}$ for $\bar{t}$ in the order in which both tuples are listed; that is, the formula asserts about $\bar x$  all atomic  $\tau$-formulas (with parameters) that hold for $\bar t$ in $\struct{P}$.  
Finally, the primitive positive formula $\prexists{(\struct{P},\bar{t})(\bar{x})}$ is obtained from $\pre{(\struct{P},\bar{t})(\bar{x})}$ by existentially quantifying over all variables which are not contained in $\bar{x}$. That is, the formula asserts the existence of a homomorphic image of the piece $(\struct{P},\bar{t})$ around $\bar x$, with $\bar x$ taking the place of $\bar t$.  

\begin{remark} The definition of a piece in the present paper is more general than the one in~\cite{hubicka2015} because we do not require $\bar t$ to be a \emph{minimal separating cut} of $\struct{P}$. 
Consequently, the arity of the relation symbols $(\struct{P},\bar{t})$ might be larger than in~\cite{hubicka2015}.  
\end{remark}
 
For a set  $\mathcal{F}$ of structures with a signature $\tau$, we denote by $\Forb_{h}(\mathcal{F})$ the class of all finite $\tau$-structures which do not admit a homomorphism from any member of $\mathcal{F}$.
The following theorem of Hubi\v{c}ka and Ne\v{s}et\v{r}il states that every class of the form $\Forb_{h}(\mathcal{F})$, where each member of $\mathcal{F}$ is connected, almost has the amalgamation property, up to taking an expansion where pieces of the structures in $\mathcal{F}$ are marked using fresh predicates $\rho$; moreover, the class of all linear-order expansions of structures from such a class  has the Ramsey property.

\begin{theorem}[Hubi\v{c}ka and Ne\v{s}et\v{r}il, Theorem~A.4 in~\cite{bodirsky2018_article}] \label{thm:hubicka_nesetril} Let $\mathcal{F}$ be a finite set of finite connected structures over a common finite relational signature $\tau$.
Let $\rho$ be the signature whose elements are the pieces $(\struct{P},\bar{t})$ of structures from $\mathcal{F}$.
Consider the class $\mathcal{K}_{\mathrm{HN}}(\mathcal{F})$ consisting of all substructures of those $\rho$-expansions of structures from $\Forb_{h}(\mathcal{F})$ that satisfy 
\begin{equation}
    \forall \bar{x} \big((\struct{P},\bar{t})(\bar{x})  \iff   \prexists{(\struct{P},\bar{t})(\bar{x})} \big). \label{eq:HN} 
\end{equation} 
 Then the class $\mathcal{K}^{<}_{\mathrm{HN}}(\mathcal{F}):=\bigl(\mathcal{K}_{\mathrm{HN}}(\mathcal{F})\bigr)^<$  has the RP and the AP.
\end{theorem}

To illustrate the construction in Theorem~\ref{thm:hubicka_nesetril}, consider the two sets $\mathcal{F}_1$ and $\mathcal{F}_2$ consisting of the canonical databases of the forbidden patterns of $\Phi_1$ and $\Phi_2$ from Example~\ref{ex:recolouring}. 
It is not hard to see that $\Forb_{h}(\mathcal{F}_2) $ has the free amalgamation property because the forbidden graphs are coloured cliques; moreover, $\Forb_{h}(\mathcal{F}_2)^{<}$ has the AP and the RP by the theorem of Ne\v{s}et\v{r}il and R\"{o}dl~\cite{NESETRIL1983183}.
Therefore, $\mathcal{F}_2$ represents a trivial case where adding $\rho$-predicates is not necessary.
This is not true for $\mathcal{F}_1$, since $\Forb_{h}(\mathcal{F}_1)$ does not have the AP. 
In Figure~\ref{fig:pieces}, we provide a graphical representation of the pieces of structures that need to be stored using $\rho$-predicates in order to obtain a class of expansions of structures from $\Forb_{h}(\mathcal{F}_1)$ with the AP and the RP.

\begin{figure}[ht]
     \centering
      \includegraphics[width=0.98\textwidth]{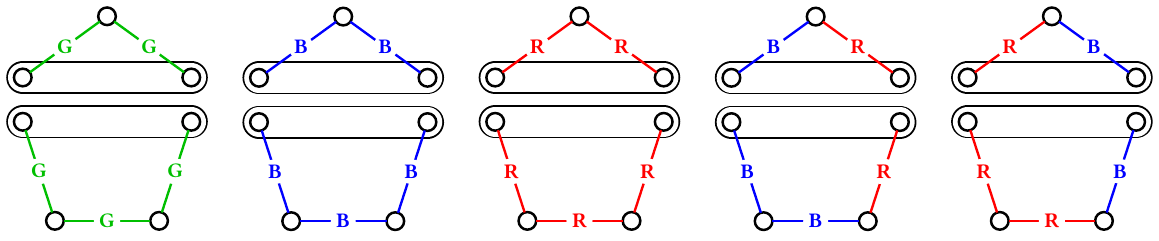}
     \caption{An illustration of some pieces of the canonical databases of the forbidden patterns of the SNP sentence $\Phi_1$ from Example~\ref{ex:recolouring}; the roots of the pieces are marked by the black lines.}
     \label{fig:pieces}
 \end{figure}

\subsection{A proof of Theorem~\ref{thm:2NEXPTIME_for_GMSNP}(\ref{item:containment})} 

Besides the reduction to the connected case, the proof of Theorem~\ref{thm:2NEXPTIME_for_GMSNP}(\ref{item:containment}) is a straightforward consequence of the combination of the following three intermediate results.
First, Lemma~\ref{lemma:recolouring_nexptime} from Section~\ref{section:recolourings} gives a (non-deterministic) 2-exponential upper bound on the time complexity of finding a recolouring between two SNP sentences.
Second, Proposition~\ref{prop:from_GMSNP_to_SNP_with_AP_and_RP} below states that every connected GMSNP sentence is equivalent to an at most 2-exponentially larger SNP sentence whose class of expanded finite models enjoys the AP and the RP. 
Third, Proposition~\ref{prop:recolouring_containment} states that, for SNP sentences with the latter property, containment is equivalent to the existence of a recolouring. 
Here we are in fact very lucky that the parameters on which  Lemma~\ref{lemma:recolouring_nexptime} and  Proposition~\ref{prop:from_GMSNP_to_SNP_with_AP_and_RP} depend 2-exponentially are different -- this will become clear in the proof of Theorem~\ref{thm:2NEXPTIME_for_GMSNP}.

 \begin{restatable}{proposition}{SNPAPRP}    \label{prop:from_GMSNP_to_SNP_with_AP_and_RP}  
 From every connected GMSNP $\tau$-sentence $\Phi$ one can compute in deterministic 2-exponential time 
 an SNP $\tau$-sentence $\Delta(\Phi)$ such that:
\begin{enumerate} 
    \item $\wh(\Delta(\Phi)) \in \mathcal{O}\bigl(\wh(\Phi)\bigr)$;
    \item $\ar(\Delta(\Phi)) \in \mathcal{O}\bigl(\ar(\Phi)+\wh(\Phi)\bigr)$;
    \item $\hh(\Delta(\Phi)) \in \mathcal{O}\bigl(\hh(\Phi)+\lh(\Phi)\cdot 2^{\wh(\Phi)}\bigr)$
    \item $\lh(\Delta(\Phi)) \in \mathcal{O}\bigl(2^{(\hh(\Phi)+\lh(\Phi))\cdot 2^{\wh(\Phi)+\ar(\Phi)}}\bigr)$;
    \item \label{item:AP_andRP} $\efm(\Delta(\Phi))$ has the AP and the RP;
    \item $\fm(\Delta(\Phi))=\fm(\Phi)$.
\end{enumerate}

\end{restatable}
 
 The proof of Proposition~\ref{prop:from_GMSNP_to_SNP_with_AP_and_RP} below builds on the proof of Theorem~5 in~\cite{bodirsky_asnp}.
\begin{proof}
Let $\Phi$ be an arbitrary connected GMSNP $\tau$-sentence.
We first obtain the signatures $\tau'$ and $\sigma'$ by introducing, for every $R\in \tau\cup \sigma$, two new symbols $R_+$ and $R_-$ of the same arity. 
Next, we obtain an auxiliary GMSNP $\tau'$-sentence $\Phi'$ from $\Phi$ as follows.
For each forbidden pattern $\phi$ of $\Phi$, sentence $\Phi'$ contains the forbidden pattern $\phi'$ obtained by replacing each positive $(\tau\cup\sigma)$-atom $R(\bar{x})$ with $R_+(\bar{x})$ and each negative $(\tau\cup\sigma)$-atom $\neg R(\bar{x})$ with $R_-(\bar{x})$. 

Let $\mathcal{F}$ be the set of canonical databases for the forbidden patterns of $\Phi'$ (which are just conjunctions of positive $\tau'\cup \sigma'$-atoms) and $\rho$ be the signature whose elements are the pieces of structures from $\mathcal F$. 
Consider the $\tau'$-sentence $\Phi''$ obtained from $\Phi'$ by including the following new clauses. 
\begin{enumerate}[label={\roman*.}] 
\item \label{item:clauses1} Clauses defining the linear order using irreflexivity, totality, and transitivity:
    \[
    \forall x,y,z  \, \neg (x<x) \wedge    \big( (x<y) \vee (x=y) \vee (y<x) \big) \wedge  \big(  (x<y) \wedge (y<z) \Rightarrow  (x<z) \big).
    \] 
    \item \label{item:clauses2} Each clause $\phi$ over 
    $\tau'\cup \sigma' \cup \rho$ of the form $\psi(\bar{x},\bar{y})\Rightarrow (\struct{P}',\bar{t}')(\bar{x})$ (for $(\struct{P}',\bar{t}')\in \rho$) with at most $\wh(\Phi)$-many variables, where $\psi$ is a conjunction of positive atoms with the following property: 
\begin{center} \vspace{0.75em}
\framebox{\parbox{0.68\textwidth}{ 
 Replacing every atomic subformula $(\struct P ,\bar t)(\bar{u})$  of $\psi$ with $\pre{(\struct P ,\bar t)(\bar{u})}$ yields a $(\tau'\cup \sigma')$-formula $\psi^{-1}$ such that there exists a homomorphism $h$ from the canonical database of $\pre{(\struct{P}',\bar{t}')(\bar{x})}$ to the canonical database of $\psi^{-1}$ with $h(\bar{x})=\bar{x}$.   
}} \vspace{0.75em}
\end{center} 
    \item \label{item:clauses3} Each clause $\phi$ over 
    $\tau'\cup \sigma' \cup \rho$ with at most $\wh(\Phi)$-many variables whose forbidden pattern is a conjunction of positive atoms with the following property: 
 \begin{center} \vspace{0.75em}
\framebox{\parbox{0.58\textwidth}{ 
 Replacing every atomic subformula  $(\struct P ,\bar t)(\bar{u})$ with $\pre{(\struct P ,\bar t)(\bar{u})}$ yields a formula whose canonical database admits a homomorphism from a member of $\mathcal{F}$. 
}} \vspace{0.75em}
\end{center} 
    \end{enumerate}

Before we proceed further, let us assess the size of $\Phi''$ in detail.
The number of variables per clause remains in $\mathcal{O}(\wh(\Phi))$.
The arity of the symbols from $\rho$ is in $\mathcal{O}\bigl(\ar(\Phi)+\wh(\Phi)\bigr)$, because $\Phi''$ inherits all symbols from $\Phi'$ and gains additional symbols whose arity is bounded by the domain size of structures in $\mathcal{F}$.
The number of relation symbols increases exponentially in $\wh(\Phi)$ and polynomially in the remaining parameters, more specifically it is in $\mathcal{O}\bigl(\hh(\Phi)+\lh(\Phi)\cdot 2^{\wh(\Phi)}\bigr)$; the reason is that the pieces of structures in $\mathcal{F}$ are generated by all possible substructures.
Regarding the number of clauses: item~\ref{item:clauses1} produces constant number of clauses and items~\ref{item:clauses2} and~\ref{item:clauses3} produce $\mathcal{O}\bigl(2^{(\hh(\Phi)+\lh(\Phi))\cdot 2^{\wh(\Phi)+\ar(\Phi)}} \bigr)$ many clauses.
The reason for the 2-exponential blow-up in item~\ref{item:clauses3} is that there are as many new clauses as there are structures of size at most $\wh(\Phi)$ over the signature of $\Phi''$.
Similarly, also item~\ref{item:clauses2} yields a 2-exponential blow-up; here the number from item~\ref{item:clauses3} is additionally multiplied with the number of symbols in $\rho$ (i.e., pieces of structures in $\mathcal{F}$), which is in $\mathcal{O}\bigl(\lh(\Phi)\cdot 2^{\wh(\Phi)}\bigr)$.

Note that all structures in $\mathcal{F}$ are connected since $\Phi$ is connected.
By Theorem~\ref{thm:hubicka_nesetril}, the class $\mathcal{K}^{<}_{\mathrm{HN}}(\mathcal{F})$ over the signature $\tau'\cup\sigma'\cup \rho \cup \{<\}$ has the AP and the RP.
 Theorem~\ref{theorem:fraisse_2} then implies that there exists a homogeneous Ramsey structure $\struct{C}$ with $\age(\struct{C})=\mathcal{K}^{<}_{\mathrm{HN}}(\mathcal{F})$.
 By the construction of $\Phi''$, we have that $\age(\struct{C})=\efm(\Phi'')$ (Claim~\ref{claim:age_coloured}); before we give a rigorous proof, we provide an intuitive explanation of why this is the case.
First, item~\ref{item:clauses1} ensures that $<$ interprets as a linear order in each $\struct{A}\in \efm(\Phi'')$ which, for lack of additional clauses with $<$, is independent of the remaining signature; the same 
 is true in each $\struct{A}\in \age(\struct{C})$.
Next, consider the remaining part $\tau'\cup \sigma'\cup\rho$ of the signature. If $\struct A\in\age(\struct C)$, then by definition, it does not admit a homomorphism from any member of $\mathcal F$; moreover,  its $\rho$-predicates are obtained through equation~\eqref{eq:HN} by definition, and hence it satisfies items~\ref{item:clauses2} and~\ref{item:clauses3}. Whence, $\struct A\in\efm{(\Phi'')}$. Conversely, if a finite structure  $\struct A$ satisfies the first-order part of $\Phi''$, then extending it by new elements so that the truth of all $\rho$-predicates is witnessed as in equation~\eqref{eq:HN} yields a structure in which the equivalence~\eqref{eq:HN} holds (because of item~\ref{item:clauses2}) and which does not admit a homomorphism from any member from $\mathcal F$ (because of item~\ref{item:clauses3}); hence $\struct A\in\age(\struct C)$.

\begin{claim} \label{claim:age_coloured}
    $\age(\struct{C})=\efm(\Phi'')$
\end{claim}
\begin{claimproof} First, we show that every finite model of the first-order part of $\Phi''$ can be embedded into $\struct{C}$. 
To this end, let $\struct{A}\in \efm(\Phi'')$ be arbitrary.
Let $\phi$ be the canonical query of the $(\tau'\cup \sigma' \cup \rho)$-reduct of $\struct{A}$, and let $\phi^{-1}$ be the formula obtained by replacing each subformula $(\struct P,\bar t)(\bar{u})$ in $\phi$ by $\pre{(\struct P,\bar t)(\bar{u})}$.
Now, let $\struct{A}'$ be the canonical database of $\phi^{-1}$, viewed as a structure over the signature $\tau'\cup \sigma'$.
Since $\struct{A}$ satisfies each clause from item~\ref{item:clauses3} for all possible variable substitutions, we have $\struct{A}'\in \Forb_h(\mathcal{F})$.
Let $\struct{A}''$ be the $(\tau'\cup \sigma' \cup \rho)$-expansion of $\struct{A}'$ where $\rho$-relations are defined using eq.~\eqref{eq:HN}.
Clearly, the substructure of $\struct{A}''$ on $A$ has all relational $\rho$-tuples of $\struct{A}$.
Since $\struct{A}$ satisfies each clause from item~\ref{item:clauses2} for all possible variable substitutions, $\struct{A}$ also has all relational $\rho$-tuples of the substructure of $\struct{A}''$ on $A$.
We conclude that the $(\tau'\cup \sigma' \cup \rho)$-reduct of $\struct{A}$ is in $\mathcal{K}_{\mathrm{HN}}(\mathcal{F})$.
Since $\struct{A}$ satisfies the sentence in item~\ref{item:clauses1}, $<$ interprets as a linear order in $\struct{A}$, and hence  $\struct{A}\in  \age(\struct{C})=\mathcal{K}^{<}_{\mathrm{HN}}(\mathcal{F})$.  

Next, we show that every element of $\age(\struct{C})$ satisfies the first-order part of $\Phi''$.
To this end, let $\struct{A}\in \age(\struct{C})=\mathcal{K}^{<}_{\mathrm{HN}}(\mathcal{F})$ be arbitrary.
Then, by Theorem~\ref{thm:hubicka_nesetril}, there exists a structure $\struct{A}'\in \Forb_h(\mathcal{F})$, a $\rho$-expansion $\struct{A}''$ of $\struct{A}'$ satisfying eq.~\eqref{eq:HN}, and a linear-order expansion $\struct{A}'''$ of $\struct{A}''$ such that $\struct{A}$ is a substructure of $\struct{A}'''$.
The clauses of $\Phi''$ stemming from $\Phi'$ are satisfied in $\struct{A}$ for all possible variable substitutions because $\struct{A}'\in \Forb_h(\mathcal{F})$.
We continue with the clauses coming from items~\ref{item:clauses1},~\ref{item:clauses2}, and~\ref{item:clauses3}.
The sentence in item~\ref{item:clauses1} clearly holds in $\struct{A}$ because $<$ interprets in $\struct{A}'''$ as a linear order on $A$. 
The structure $\struct{A}$ satisfies each clause from item~\ref{item:clauses2} for all possible variable substitutions because $\struct{A}''$ satisfies eq.~\eqref{eq:HN}.
For item~\ref{item:clauses3}, suppose, on the contrary, that there exists a clause $\phi$ whose forbidden pattern is a conjunction of positive atoms with the property as in item~\ref{item:clauses3} and a tuple $\bar{a}$ over $A$ such that $\struct{A}\not\models \phi(\bar{a})$.
Since $\struct{A}''$ is defined from $\struct{A}'$ using eq.~\eqref{eq:HN}, the structure $\struct{A}'$ admits a homomorphism from a member of $\mathcal{F}$ (witnessed by the tuple $\bar{a}$ as well as the existential witnesses for eq.~\eqref{eq:HN}). 
This is a contradiction to $\struct{A}'\in \Forb_h(\mathcal{F})$, hence the clauses in item~\ref{item:clauses3} hold for $\struct{A}$ as well (for all possible variable substitutions). 
\end{claimproof} 

We say that a set $S\subseteq C$ is \emph{correctly labelled} if every tuple $\bar{t}$ over $S$ matching the arity of some $X\in \sigma$ is contained either in $X_+^{\struct{C}}$ or in $X_-^{\struct{C}}$, but not both.
Consider the $(\tau\cup\sigma\cup \tau'\cup \sigma' \cup \rho \cup \{<\})$-expansion $\struct{C}_{\Phi}$ of $\struct{C}$ where each $R\in \tau\cup \sigma$ of arity $k$ interprets as the $k$-ary relation defined by the quantifier-free first-order formula
\begin{align}
R_+(x_1,\dots,x_k) \wedge \bigl(\{x_1,\dots, x_k\} \textit{ is correctly labelled }\bigr).  \label{eq:correctly_labeled}
\end{align} 
 As first-order definable relations are preserved by automorphisms, we have that $\struct{C}_{\Phi}$ is homogeneous Ramsey  because $\struct{C}$ is homogeneous Ramsey.

Finally, we obtain the sentence $\Delta(\Phi)$: by additionally defining $(\tau\cup \sigma)$-relations using eq.~\eqref{eq:correctly_labeled} and existentially quantifying over all symbols in $\tau'\cup \sigma$.
The length of $\Delta(\Phi)$ only increases by an additional multiplicative single exponential because we need to consider all subsets of $\bigl[\ar(\Phi)+\wh(\Phi)\bigr]$ for every relation symbol in $\sigma$; its width, height, and arity remain unchanged.
It follows from Theorem~\ref{thm:hubicka_nesetril} that $\efm(\Delta(\Phi))$ has the AP and the RP. 
Moreover, we have $\age(\struct{C}_{\Phi})=\efm(\Delta(\Phi))$.
It remains to show that $\fm(\Phi)=\fm(\Delta(\Phi))$.  
This is a direct consequence of the following claim.

\begin{claim}\label{claim:appendix}
    $\CSP(\struct{C}_{\Phi}
^{\tau})=\fm(\Phi)=\age(\struct{C}_{\Phi}
^{\tau})$.
\end{claim}
\begin{claimproof} 
By definition, we have $\age(\struct{C}_{\Phi}
^{\tau}) \subseteq \CSP(\struct{C}_{\Phi}
^{\tau})$.

Next, we prove $\CSP(\struct{C}_{\Phi}
^{\tau})\subseteq \fm(\Phi)$.
Suppose that $\struct{A}\in \CSP(\struct{C}_{\Phi}
^{\tau})$, i.e., there exists a homomorphism $h\colon \struct{A}\rightarrow \struct{C}_{\Phi}
^{\tau}$. 
We define $\smash{\overline{\struct{A}}}$ as the expansion of $\struct{A}$ to the signature of $\struct{C}_{\Phi}$ obtained by pulling back the relations from $\struct{C}_{\Phi}$ through $h$.
By definition, $h$ is a homomorphism from $\smash{\overline{\struct{A}}}$ to $\struct{C}_{\Phi}$.
Now, let $\neg \big( \psi_1\wedge \cdots \wedge  \psi_k \wedge \neg \psi_{k+1} \wedge \cdots \wedge \neg \psi_{\ell}\big)$ be a clause in $\Phi$, where each $\psi_i$ is a positive atomic formula.
Suppose that, for some tuple $\bar{a}$  over $A$, we have $\smash{\overline{\struct{A}}}\models \psi_1\wedge \cdots \wedge \psi_k(\bar{a})$.
Since $h$ is a homomorphism, it follows that $$\struct{C}_{\Phi}\models \psi_1\wedge \cdots \wedge \psi_k(h(\bar{a})).$$ 
For every atom $\psi$ of the form $R(\bar{x})$ for $R\in \tau\cup \sigma$, we define $\psi^+$ as $R_+(\bar{x})$ and $\psi^-$ as $R_-(\bar{x})$.
By the definition of $\struct{C}_{\Phi}$, more specifically by~eq.~\eqref{eq:correctly_labeled}, we have 
$$\struct{C}_{\Phi}\models \psi^+_1\wedge \cdots \wedge \psi^+_k(h(\bar{a})).$$
Since $\mathcal{F}$ contains a structure associated with the forbidden pattern $\psi^+_1\wedge \cdots \wedge  \psi^+_k \wedge \psi^-_{k+1} \wedge \cdots \wedge  \psi^-_{\ell}$ and $\age(\struct{C})=\mathcal{K}^{<}_{\mathrm{HN}}(\mathcal{F})$, we must have 
$$\struct{C}_{\Phi}\models \neg \psi^-_{k+1}\vee \cdots \vee \neg \psi^-_{\ell}(h(\bar{a})),$$ i.e., there exists $i\in [\ell]\setminus [k]$
such that $\struct{C}_{\Phi}\models \neg \psi^-_{i}(h(\bar{a}))$.
Since $\Phi$ satisfies the guarding axiom, there exists $j \in [k]$ such that all variables $\bar{x}_i$ of $\psi_{i}$ are among the variables $\bar{x}_j$ of $\psi_{j}$.
By the assumption that the set of all entries in $h(\bar{x}_j)$ is correctly labelled, we have $\struct{C}_{\Phi}\models \psi^+_{i}(h(\bar{a}))$. 
As a subset of a correctly labelled set, the set of all entries in  $h(\bar{x}_i)$ is also correctly labelled.
Hence, by definition we have $\struct{C}_{\Phi}\models \psi_{i}(h(\bar{a}))$. 
Since $\Phi$ is monotone, the atomic formula $\psi_{i}$ is of the form $X(\bar{x})$ for $X\in \sigma$.
Hence, by the definition of $\smash{\overline{\struct{A}}}$, we have $\smash{\overline{\struct{A}}}\models  \psi_{i}(h(\bar{a}))$.
This means that $\smash{\overline{\struct{A}}}
^{\tau\cup \sigma}$ provides the desired expansion of $\struct{A}$ in $\efm(\Phi)$.

Finally, we show $\fm(\Phi)\subseteq \age(\struct{C}_{\Phi}
^{\tau})$.
Suppose that $\struct{A}\models \Phi$, and let $\struct{A}_{\Phi}$ be a $\tau\cup \sigma$-expansion of $\struct{A}$ witnessing this fact. 
We define $\struct{A}'$ as the $(\tau'\cup \sigma')$-structure with domain $A$ and where $R_+$ interprets as $R^{\struct{A}_{\Phi}}$ and $R_-$ as the complement of $R^{\struct{A}_{\Phi}}$ in $\struct{A}_{\Phi}$.
Then clearly $\struct{A}'\in \efm(\Phi')$.
Since the relations of $\struct{A}'$ satisfy the correct labelling condition,  $\struct{A}'$ is a reduct of a substructure $\smash{\overline{\struct{A}}}$ of $\struct{C}_{\Phi}$.
Clearly, the $\tau$-reduct of $\smash{\overline{\struct{A}}}$ is isomorphic to $\struct{A}$.
It follows that $\struct{A}\in \age(\struct{C}_{\Phi}
^{\tau})$.  
\end{claimproof}

This finishes the proof.
\end{proof}

The next proposition shows that the property in item~\ref{item:AP_andRP} of Proposition~\ref{prop:from_GMSNP_to_SNP_with_AP_and_RP} attained through the transformation $\Delta$ is sufficient for a reduction from containment to recolouring.
 
 \begin{proposition}    \label{prop:recolouring_containment}
  For SNP $\tau$-sentences $\Phi_1$ and $\Phi_2$ such that $\efm(\Phi_1)$ and $\efm(\Phi_2)$ both have the AP and the RP  the following are equivalent.
    \begin{enumerate}
        \item \label{item:recolouring1} $\fm(\Phi_1)\subseteq \fm(\Phi_2)$.
        \item \label{item:recolouring2} There exists a recolouring from $\Phi_1$ to $\Phi_2$.
    \end{enumerate}  
\end{proposition}    
\begin{proof} Let $n\coloneqq \max(\ar(\Phi_1),\ar(\Phi_2))$.

``(\ref{item:recolouring1})$\Rightarrow$(\ref{item:recolouring2})''   
By Theorem~\ref{theorem:fraisse_2}, there exist homogeneous Ramsey structures $\struct{C}_1$ and $\struct{C}_2$ such that $\efm(\Phi_1)=\age(\struct{C}_1)$ and $\efm(\Phi_2)=\age(\struct{C}_2)$.
Then, $\age(\struct{C}_1
^{\tau}) \subseteq  \age(\struct{C}_2
^{\tau})$.
By the $\omega$-categoricity of $\struct{C}_2
^{\tau}$,  Lemma~\ref{lemma:compactness} implies that there exists an embedding $e$ from $\struct{C}_1
^{\tau}$ to $\struct{C}_2
^{\tau}$. 
By Theorem~\ref{th:canonical_ramsey}, we may assume that $e$ is canonical as a function from $\struct{C}_1$ to $\struct{C}_2$ since the former is an $\omega$-categorical Ramsey structure.
Let $\xi'$ be the mapping from $\efm(\Phi_1)$ to $\efm(\Phi_2)$ defined as follows.
Given $\struct{A}\in \efm(\Phi_1)$, consider an arbitrary embedding $i$ from $\struct{A}$ to $\struct{C}_1$.
Then, we define $\xi'(\struct{A}) \in \efm(\Phi_2)$ with domain $A$ whose relations are defined through their preimages under $e\circ i$.   
Let $\xi$ be the restriction of $\xi'$ to $\colours(\Phi_1)$.
We claim that $\xi'$ arises from $\xi$ as in the definition of a recolouring. 

First, the $\tau$-reduct of any   $\struct{T}\in\colours(\Phi_1)$  and its $\xi$-image are identical since $\xi(\struct{T})=\xi'(\struct{T})$ is obtained by pulling back relations via the $\tau$-embedding $e\circ i$.  
For the second part of the definition pertaining to $\xi'$, let $\struct A\in\efm{(\Phi_1)}$, $\struct{T}\in\colours(\Phi_1)$,  and $f\in \binom{\struct A}{\struct{T}}$
be arbitrary; we claim that $f\in \smash{\binom{\xi'(\struct A)}{\xi(\struct{T})}}$. By definition, the relations of $\xi'(\struct A)$ and of $\xi(\struct{T})$ are obtained by first embedding $\struct A$ and $\struct{T}$  into $\struct{C}_1$, then applying $e$, and then pulling back the relations from $\struct{C}_2$. This definition does not depend on the choice of the embeddings into $\struct{C}_1$, by the homogeneity of  $\struct{C}_1$ and the canonicity of $e$ -- the claim follows. 

``(\ref{item:recolouring2})$\Rightarrow$(\ref{item:recolouring1})''
Let $\struct{A}\in \fm(\Phi_1)$ be arbitrary; then there exists an expansion $\struct{A}'\in \efm(\Phi_1)$. By the definition of a recolouring we have $\xi'(\struct{A}')\in  \efm(\Phi_2)$. It also follows from this definition that  for all $\struct{T} \in \colours(\Phi_1)$ and every embedding $e\in \smash{\binom{\struct{A}}{\struct{T}}}$ we have $e\in \binom{\xi'(\struct{A})}{\xi(\struct{T})}$. By our choice of $n$, this implies that the $\tau$-reduct of $\xi'(\struct{A}')$ coincides with $\struct A$. Hence $\xi'(\struct{A}')$ witnesses that  $\struct{A}\in \fm(\Phi_2)$.
\end{proof}

\begin{proof}[Proof of Theorem~\ref{thm:2NEXPTIME_for_GMSNP}(\ref{item:containment})] 
 By Proposition~\ref{cor:connected_matters_not}, we can restrict ourselves to connected GMSNP.
 Now the statement follows directly from Proposition~\ref{prop:from_GMSNP_to_SNP_with_AP_and_RP} combined with Proposition~\ref{prop:recolouring_containment} and Lemma~\ref{lemma:recolouring_nexptime}.
 To see this, note that, through the application of Proposition~\ref{prop:from_GMSNP_to_SNP_with_AP_and_RP}, the arity and width increase polynomially, the height increases 1-exponentially, and the length increases 2-exponentially.
 According to Lemma~\ref{lemma:recolouring_nexptime}, the complexity of testing the existence of a recolouring only depends 1-exponentially on the height, so a 1-exponential increase in height is not an issue.
 Moreover, the complexity only depends polynomially on the length, so a 2-exponential increase in length is not an issue either.
 We conclude that the total complexity is indeed $\TWONEXPTIME$.
\end{proof}

\subsection{A proof of Theorem~\ref{thm:2NEXPTIME_for_GMSNP}(\ref{item:FO_rewritability})}  \label{section:fo_rewritability}
 
Let $\Phi$ be a connected GMSNP $\tau$-sentence, and let $\Delta(\Phi)$ be the SNP $\tau$-sentence from Proposition~\ref{prop:from_GMSNP_to_SNP_with_AP_and_RP}.
We consider the reduction from the model-checking problem for $\Delta(\Phi)$ to a finite-domain CSP originating from~\cite{bodirsky2016reducts,BodirskyM18}.
We define $\struct{T}_{\Phi}$ to be the following relational structure.
Let $n\coloneqq \max\bigl(\ar(\Delta(\Phi))+1,\wh(\Delta(\Phi)),3\bigr).$
The domain of $\struct{T}_{\Phi}$ is the set $\colours(\Delta(\Phi))$, and the signature of $\struct{T}_{\Phi}$ contains:
\begin{itemize}
    \item for every $R\in \tau$ of arity $m$ and every $g\colon [m] \rightarrow [n]$, a unary symbol $U_{R,g}$;
    \item for every $m\in [n]$ and all $f,g\colon [m] \rightarrow [n]$, a binary symbol $E_{f,g}$.
\end{itemize}

 Each unary symbol $U_{R,g}$ interprets as the set of all $\struct{C}\in \colours(\Delta(\Phi))$ for which $\big(g(1),\dots,g(m)\big)\in  R^{\struct{C}}$, and each binary symbol $E_{f,g}$ interprets as the set of all pairs $(\struct{C}_1,\struct{C}_2)\in \colours(\Delta(\Phi))^2$ where the map $f(i) \mapsto g(i)$ is a well-defined partial isomorphism from $\struct{C}_1$ to $\struct{C}_2$. 
 
 The following result was essentially obtained in~\cite[Thm.~3.1]{BodirskyM18}, except that the language used there was slightly different; see the remark at the end of~\cite[Sec.~3]{BodirskyM18}.
 \begin{theorem}[Thm.~3.1 in~\cite{BodirskyM18}] \label{theorem:FO_red_to_orbits} There exists a FO-reduction from the model-checking problem for $\Delta(\Phi)$ to $\CSP(\struct{T}_{\Phi})$. 
 \end{theorem} 
 The formulation of Theorem~3.1 in~\cite{BodirskyM18} corresponds to the situation where there exist relational structures $\struct{A}$ and $\struct{B}$ such that $\CSP(\struct{A})=\age(\struct{A})=\fm(\Delta(\Phi))$ and $\age(\struct{B})=\efm(\Delta(\Phi))$.
 We know that such structures exist by the proof of Proposition~\ref{prop:from_GMSNP_to_SNP_with_AP_and_RP}; take $\struct{B}\coloneqq \struct{C}_{\Phi}$ and $\struct{A} \coloneqq \struct{C}_{\Phi}^{\tau}$. 
Note that, since $\struct{C}_\Phi$ is homogeneous and $n \geq \ar(\Delta(\Phi))$, the set $\colours(\Delta(\Phi))$ is in a 1-1 correspondence with the $n$-ary orbits of $\struct{C}_\Phi$.

Recall the definition of the induced map $f^{\acts}$ on orbits for canonical polymorphisms $f$ of $\struct{C}_\Phi^\tau$ from the last paragraph of Section~\ref{section:prelims_structures}.

\begin{theorem}[Lem.~4.7 in~\cite{BodirskyM18}]\label{th:pols_of_orbit_structure}
    For every $d$-ary  polymorphism $f$ of $\struct{C}_\Phi^\tau$ canonical with respect to $\struct{C}_\Phi$, the operation $f^{\acts}\colon T_\Phi^d\to T_\Phi$ is a polymorphism of $\struct{T}_\Phi$. 
\end{theorem} 

The following definition originates from~\cite{larose2007characterisation}.
The \emph{one-tolerant $n$-th power} ${}^1\struct{A}^n$ of a $\tau$-structure $\struct{A}$ is the $\tau$-structure with domain $A^n$ and where every symbol $R\in\tau$ of arity $k$ interprets as the relation 
\[
 \Big\{
\bigl((a_1^1,\ldots,a_1^n),\dots,(a_k^1,\ldots,a_k^n)\bigr)\in (A^n)^k \ \Big| \  \bigl|\{j\in [n]\mid (a_1^j,\ldots,a_k^j)\in R^\struct{A}\}\bigr| \geq n-1 
\Big\}.  
\]  

For every $\omega$-categorical structure $\struct{A}$, the membership of $\CSP(\struct{A})$ in FO is characterised by the existence of a \emph{one-tolerant} polymorphism of $\struct{A}$, i.e., a homomorphism $f\colon {}^1\struct{A}^n\to \struct{A}$, for some $n\in\mathbb{N}$.
This was essentially shown in~\cite{bodirsky_hils_martin}, extending a result from~\cite{larose2007characterisation}.
However, the authors of that article opted for a more general formulation characterising the membership in FO for all structures in a finite signature, even structures which are not necessarily  $\omega$-categorical.
We now extract the result tailored to the $\omega$-categorical setting.

\begin{proposition}[cf.~Theorem~5.7 in~\cite{bodirsky_hils_martin}] \label{prop:bhm12}
    For every $\omega$-categorical structure $\struct{A}$, the following are equivalent:
    \begin{enumerate} 
    \item\label{item:bhm12_1} $\CSP(\struct{A})$ is in $\mathrm{FO}$.
    \item\label{item:bhm12_2} $\struct{A}$ has a one-tolerant polymorphism.
    \end{enumerate} 
\end{proposition}
\begin{proof}
We start with an auxiliary claim.
Let $\tau$ be the signature of $\struct{A}$.
We call a finite $\tau$-structure $\struct{B}$ a \emph{critical obstruction} for $\CSP(\struct{A})$ if $\struct{B}\not\to\struct{A}$ but removing any relational tuple from $\struct{B}$ yields a structure that maps homomorphically to $\struct{A}$. 
\begin{claim}    \label{claim:bhm12}
    For $n\in \mathbb{N}$, the following are equivalent:
    \begin{enumerate} 
    \item\label{item:bhm12_1a} $\struct{A}$ has an $(n+1)$-ary one-tolerant polymorphism.
    \item\label{item:bhm12_2a} Every critical obstruction to $\CSP(\struct{A})$ has at most $n$ relational tuples.
    \end{enumerate}
\end{claim}
\begin{claimproof}
    ``(\ref{item:bhm12_1a})$\Rightarrow$(\ref{item:bhm12_2a})''.
    Let $f$ be an $(n+1)$-ary one-tolerant polymorphism of $\struct{A}$. Suppose, on the contrary, that there exists a critical obstruction $\struct{B}$ for $\CSP(\struct{A})$ with $m>n$ relational tuples $\bar b_1,\ldots,\bar b_m$.
    We define a new $m$-ary polymorphism $f'$ by $$f'(x_1,\dots,x_m)\coloneqq f(x_1,\dots, x_{n+1}).$$ 
    It is easy to see that also $f'$ is one-tolerant.
    For every $i\in [m]$, let $\struct{B}_i$ be obtained from $\struct{B}$ by removing $\bar b_i$ from the corresponding relation and let $h_i\colon \struct{B}_i\to\struct{A}$ be any homomorphism (which exists by assumption).
    Then, the map $h = (h_1,\ldots,h_m)$ is a homomorphism from $\struct{B}$ to ${}^1\struct{A}^{m}$.
    But then $f'\circ h\colon \struct{B}\to\struct{A}$ is a homomorphism, contradicting $\struct{B}$ being an obstruction.
    Hence, the statement of (\ref{item:bhm12_2a}) holds.

    ``(\ref{item:bhm12_2a})$\Rightarrow$(\ref{item:bhm12_1a})''
    Since ${}^1\struct{A}^{n+1}$ is countable, by Lemma~\ref{lemma:compactness}, $\struct{A}$ has an $n$-ary one-tolerant polymorphism if and only if every finite substructure of ${}^1\struct{A}^{n+1}$ maps homomorphically to $\struct{A}$.
    Suppose, on the contrary, that there is no $(n+1)$-ary one-tolerant polymorphism of $\struct{A}$, then some finite substructure $\struct{B}$ of ${}^1\struct{A}^{n+1}$ does not map to $\struct{A}$.
    Let $\struct{B}'$ be any critical obstruction for $\CSP(\struct{A})$ obtained by removing relational tuples from $\struct{B}$.
    For every $i\in [n+1]$, let $\pi_i^{n+1}$ be the $i$-th $(n+1)$-ary projection on this structure. None of these projections can be a homomorphism when viewed as a mapping from $\struct{B}'$ into $\struct{A}$;  thus, for every $i$ there exists $R\in \tau$ and $\bar b_i\in R^{\struct{B}'}$ such that $\pi^{n+1}_i(\bar b_i) \notin R^{\struct{A}}$. 
    By the definition of ${}^1\struct{A}^{n+1}$, it must be the case that $\pi^{n+1}_j(\bar b_i) \in R^{\struct{A}}$ for every $j \in [n+1]\setminus \{i\}$.
    Thus, we have  $\bar b_i\not=\bar b_j$ whenever $i\neq j$.
    But this contradicts $\struct{B}'$ having at most $n$ relational tuples. Hence, the statement of (\ref{item:bhm12_1a}) holds.
\end{claimproof}

    ``(\ref{item:bhm12_1})$\Rightarrow$(\ref{item:bhm12_2})''
    If $\CSP(\struct{A})$ is definable in $\mathrm{FO}$, then, similarly to the proof of Claim~\ref{claim:connected2}, $\CSP(\struct{A})$ is equivalent to a universal negative first-order sentence that can be written in the form $\bigwedge_i \forall \bar x_i\ldotp\neg\phi_i(\bar x_i)$, where each $\phi_i$ is a conjunction of non-negated $\tau$-atoms.
    This implies that the number of relational tuples in every critical obstruction to $\CSP(\struct{A})$ is bounded by the maximum of the number of $\tau$-atoms that appears within one of the formulas $\phi_i$.
    Then, by Claim~\ref{claim:bhm12}, $\struct{A}$ has a one-tolerant polymorphism of every arity exceeding that number.
    
    ``(\ref{item:bhm12_2})$\Rightarrow$(\ref{item:bhm12_1})''.
    By Claim~\ref{claim:bhm12}, if there is a one-tolerant polymorphism of arity $n+1$, then every critical obstruction has at most $n$ tuples. 
    This implies that there are only finitely many critical obstructions, from which it easily follows that $\CSP(\struct{A})$  is definable in  FO.
\end{proof}

 We now have all the ingredients to prove Theorem~\ref{thm:2NEXPTIME_for_GMSNP}(\ref{item:FO_rewritability}). 
\begin{proof}[Proof of Theorem~\ref{thm:2NEXPTIME_for_GMSNP}(\ref{item:FO_rewritability})]  
 By Proposition~\ref{cor:connected_matters_not}, we can restrict ourselves to connected GMSNP.
 Given a connected GMSNP sentence $\Phi$, we first compute $\Delta(\Phi)$ from Proposition~\ref{prop:from_GMSNP_to_SNP_with_AP_and_RP}.
 Then, we compute the finite structure $\struct{T}_{\Phi}$ as defined above Theorem~\ref{theorem:FO_red_to_orbits}.
 This can be done in deterministic 2-exponential time, similarly as $\Delta(\Phi)$, e.g., the estimate on $|T_{\Phi}|$ coincides with the estimate on $\lh(\Delta(\Phi))$.
 \begin{claim}
 $\Phi$ is FO-rewritable if and only if $\CSP(\struct{T}_{\Phi})$ is definable in FO.
 \end{claim}
 \begin{claimproof}  
``$\Rightarrow$'':  As $\fm(\Delta(\Phi))=\CSP(\struct{C}_\Phi^\tau)$ and as $\struct{C}_\Phi^\tau$ is $\omega$-categorical, by Proposition~\ref{prop:bhm12} there exists a one-tolerant $m$-ary polymorphism $f'$ of $\struct{C}_\Phi^\tau$ for some $m\in \mathbb{N}$.
    By Theorem~\ref{th:canonical_ramsey_higher_ary}, there exists an $m$-ary polymorphism 
     \begin{equation}
    f\in \overline{\{ \beta  \circ  f' \circ  (\alpha_1,\dots,\alpha_m)  \mid  \alpha_1,\dots,\alpha_m,\beta \in \Aut(\struct{C}_\Phi)\}}. \label{eq:closure2}
    \end{equation} 
    which is moreover canonical with respect to $\struct{C}_\Phi$.  It follows directly from eq.~\eqref{eq:closure2} that $f$ is also one-tolerant.  
    By Theorem~\ref{th:pols_of_orbit_structure}, we have that $f^{\acts}$ is a polymorphism of $\struct{T}_\Phi$.
    Again, it is easy to see that $f^{\acts}$ is also one-tolerant.
    By Proposition~\ref{prop:bhm12}, we have that $\CSP(\struct{T}_\Phi)$ is definable in FO.
 
     ``$\Leftarrow$'': This direction follows immediately from Theorem~\ref{theorem:FO_red_to_orbits} and Lemma~\ref{lemma:FO_reduction}
 \end{claimproof}
 \medskip 
 
 To conclude the proof, we use the fact that, by~\cite[Theorem~6.1]{larose2007characterisation}, the definability in FO for finite-domain CSPs is in $\NP$ when the CSP is given in terms of the parametrizing structure.
 In our case, the size of $\struct{T}_{\Phi}$ is 2-exponential in the size of the input GMSNP sentence $\Phi$, which places the problem into $\TWONEXPTIME$.
\end{proof}

\section{A deep dive into recolourings} \label{section:recolouring_ready}
In this section, we introduce recolouring-readiness for GMSNP and give a full proof of Theorem~\ref{thm:recolouring_readiness}.
We start with weak recolouring-readiness (Section~\ref{sect:weak_recolour}), which provides a GMSNP-counterpart for the ``normal form'' for MMSNP sentences from~\cite{bodirsky2018_article}.
In Section~\ref{sect:weak_edge}, we explain how to obtain weakly recolouring-ready GMSNP sentences and elaborate on our claim that weak recolouring-readiness can be paired with a notion of a recolouring, called ``weak edge-recolouring'', generalising the corresponding notion for MMSNP in a fashion that is: (i) less abstract than the notion introduced in Section~\ref{section:recolourings} and (ii) closer to the idea of edge-recolourings as in eq.~\eqref{eq:edge_recolouring} of Section~\ref{sec:contributions}. 
The only essential difference between edge-recolourings and weak edge-recolourings is that the latter may depend on an external linear ordering of the edges.
In Section~\ref{subsec:order_important}, we provide an example showing that the dependence on an external linear ordering of the edges cannot be avoided unless we strengthen weak recolouring-readiness. 
Finally, in Section~\ref{sect:recolouring_readiness}, we formulate the full notion of recolouring-readiness and give a proof of Theorem~\ref{thm:recolouring_readiness}.

\subsection{ Weak recolouring-readiness} \label{sect:weak_recolour}
A relational structure $\struct{B}$ whose signature contains $\tau$ is \emph{$\tau$-edge-homogeneous} if it satisfies the homogeneity condition restricted to tuples $\bar{t}_1,\bar{t}_2$ contained in  $R^{\struct{B}}$ for some $R\in \tau$. 
\begin{remark}
    We remark that $\tau$-edge-homogeneity can be viewed as a natural generalisation of \emph{1-homogeneity} of Bodirsky, Madelaine, and Mottet~\cite{bodirsky2018_article} from vertices to relational $\tau$-tuples (up to a small syntactic modification as in Section~\ref{sec:intro_one}).
\end{remark} 
\begin{remark} Note that the $\tau$-reduct of a $\tau$-edge-homogeneous structure is not necessarily homogeneous.
To see this, consider the disjoint union $\struct{I}+\struct{R}$ of the countably infinite independent set of vertices $\struct{I}$ and the \emph{random graph} $\struct{R}$ (the signature is $\{E\}$). 
For $\tau=\{E\}$, the $\tau$-reduct of $\struct{I}+\struct{R}$ is $\tau$-edge-homogeneous but not homogeneous.
Indeed, all $E$-edges of $\struct{I}+\struct{R}$ are contained in $R$, and hence every partial homomorphism between $E$-edges extends to an automorphism of $\struct{R}$, which naturally extends to an automorphism of $\struct{I}+\struct{R}$ that fixes $I$.
On the other hand, every vertex from $R$ is connected to some other vertex in $\struct{I}+\struct{R}$ via an $E$-edge (due to the homogeneity of $\struct{R}$), while vertices from $I$ do not have any neighbours in $\struct{I}+\struct{R}$.
Hence, the homogeneity condition is not satisfied in $\struct{I}+\struct{R}$ already for tuples of arity $1$ (vertices). 
\end{remark}

Recall that, for a GMSNP sentence $\Phi$, the classes $\fm(\Phi)$ and $\efm(\Phi)$ consist of the finite models of $\Phi$ and the first-order part of $\Phi$, respectively.
 We say that a GMSNP $\tau$-sentence $\Phi$ is \emph{weakly recolouring-ready} if there exists a $\tau$-edge-homogeneous $\omega$-categorical Ramsey structure $(\struct{C}_{\Phi},<)$  such that:
\begin{enumerate}[label={\roman*.}] 
\item \label{item:1} $\age(\struct{C}_{\Phi},<)=\age(\struct{C}_{\Phi})^{<}$; 
    \item \label{item:2} $\age(\struct{C}_{\Phi})=\efm(\Phi)$; 
    \item \label{item:3}  $\age(\struct{C}_{\Phi}^{\tau})=\CSP(\struct{C}_{\Phi}^{\tau})=\fm(\Phi)$.  
\end{enumerate} 

\begin{theorem}    \label{thm:recolouring_readiness2}   
   For every connected GMSNP $\tau$-sentence $\Phi$,  one can construct a connected weakly recolouring-ready   
 GMSNP $\tau$-sentence $\Omega(\Phi)$ such that: 
\begin{enumerate} 
 \item \label{item:ready1} $\fm(\Omega(\Phi))=\fm(\Phi)$;
 \item \label{item:ready2} $\wh(\Omega(\Phi))\in\mathcal{O}(\wh(\Phi))$;
 \item \label{item:ready3} $\ar(\Omega(\Phi))\in\mathcal{O}(\ar(\Phi))$;
 \item \label{item:ready4} $\hh(\Omega(\Phi))\in\mathcal{O}\bigl(2^{(\wh(\Phi)+\ar(\Phi))^{\ar(\Phi)}\hh(\Phi)}\bigr)$;
 \item \label{item:ready5} $\lh(\Omega(\Phi))\in\mathcal{O}\bigl(2^{2^{(\wh(\Phi)+\ar(\Phi))\smash{^{\ar(\Phi)}}\cdot\hh(\Phi)}}\bigr)$.
\end{enumerate}   
 Moreover, $\Omega(\Phi)$ can be computed from $\Phi$ in nondeterministic 3-exponential time.

\end{theorem}  

\begin{proof} We divide the proof into six clearly labelled paragraphs (Steps~1--6).

\smallskip{\emph{Step 1: Colour complements.}} We first obtain the signature $\cplmt{\sigma} \supseteq \sigma$ by introducing, for every $X\in \sigma$, a new symbol $\cplmt{X}$ of the same arity.  
 Then we replace, for every $X\in \sigma$, each positive atom $X(\bar{x})$ in each clause of $\Phi$ by the negative atom $\neg \cplmt{X}(\bar{x})$. 
 The GMSNP sentence $\cplmt{\Phi}$ is then obtained by moreover including clauses of the form 
 \begin{equation}
     \neg \big(  R(\bar{x}) \wedge  \neg X(\bar{y}) \wedge \neg  \cplmt{X}(\bar{y})\big) \qquad \text{and} \qquad  \neg \big(   R(\bar{x}) \wedge X(\bar{y}) \wedge \cplmt{X}(\bar{y})\big) \label{eq:correctly_labeledly_labeled}
 \end{equation}  
 for all symbols $R\in \tau\cup \overline{\sigma}$ and $X\in \sigma$, where $\bar{y}\subseteq \bar{x}$ are tuples of fresh variables matching the arities of $R$ and $X$.  
 Since $\Phi$ satisfies the guarding axiom, we have $\fm(\cplmt{\Phi})=\fm(\Phi)$. 
 Note that $\lh(\cplmt{\Phi})\in \mathcal{O}(\lh(\Phi))$, $\ar(\cplmt{\Phi})=\ar(\Phi)$, $\wh(\cplmt{\Phi})=\wh(\Phi)$, and $\hh(\cplmt{\Phi})=2\cdot\hh(\Phi)$.

\smallskip{\emph{Step 2: Correct labellings.}}
 Let $\struct B$ be a structure over a relational signature containing $\tau\cup \cplmt{\sigma}$.
 We call a tuple $\bar{s}$ over $B$ $(\tau\cup \cplmt{\sigma})$-\emph{guarded} in $\struct{B}$ if there exists $R\in \tau\cup \cplmt{\sigma}$ and $\bar{t}\in R^{\struct{B}}$ such that $\bar{s}\subseteq \bar{t}$.
 We call a structure $\struct B$  \emph{correctly labelled} if every tuple $\bar{s}$ over $B$ $(\tau\cup \cplmt{\sigma})$-\emph{guarded} in $\struct{B}$ and matching the arity of some $X\in \sigma$ is contained either in $X^{\struct{B}}$ or in $\cplmt{X}^{\struct{B}}$, but not both. 
Note that,  before the addition of clauses of the form~\eqref{eq:correctly_labeledly_labeled} 
to $\cplmt{\Phi}$, the forbidden patterns were conjunctions of positive $(\tau\cup \cplmt{\sigma})$-atoms.
Let $\cplmt{\mathcal{F}}$ be the set of all correctly labelled $(\tau\cup\cplmt{\sigma})$-structures that can be obtained from the canonical databases of such  forbidden patterns by adding tuples to $\cplmt{\sigma}$-predicates. 

\smallskip{\emph{Step 3: Guarded pieces.}} Before applying  Theorem~\ref{thm:hubicka_nesetril} essentially to $\cplmt{\mathcal{F}}$, we will view the structures therein formally in an extended signature as follows.
 For every piece $(\struct P,\bar t)$ of a structure $\struct S$ in $\cplmt{\mathcal{F}}$, consider all correctly labelled extensions $\struct S'$ of $\struct S$ by the elements of a tuple $\bar s\supseteq \bar t$ such that all elements of $\bar s\setminus \bar t$ are fresh and  such that $R(\bar s)$ holds for some $R\in \tau$. 
 For every such extension $\struct S'$, we add its piece $(\struct P\cup\bar s,\bar s)$ to the  signature  $\sigma_+$, view the structures in $\cplmt{\mathcal{F}}$ as ($\tau\cup\cplmt{\sigma}\cup\sigma_+$)-structures,  and say that the original piece $(\struct P,\bar t)$ is \emph{guarded} by $(\struct P\cup\bar s,\bar s)$ in $\sigma_+$.  
 Note that we do not add any structures to $\cplmt{\mathcal{F}}$. 
 The size $|\sigma_+|$ is less than the number of $(\tau\cup\cplmt{\sigma})$-structures on at most $M+k$ elements, where $M\coloneqq\wh(\Phi)$ and $k\leq\ar(\Phi)$ is the maximal arity of a $\tau$-relation.
So, $|\sigma_+|$ has an upper bound $(M+k)\cdot 2^{(M+k)^k\hh(\cplmt{\Phi})}$.

\smallskip{\emph{Step 4: The amalgamation and the Ramsey property.}}
 Let $\mathcal G$ be the set of all structures $\struct{S}$ with domain $\bar s$ for which there exists a unique $R\in\tau$ and $(\struct P\cup\bar s,\bar s)\in \sigma_+$ as above such that $R(\bar s)$ and $(\struct P\cup\bar s,\bar s)(\bar s)$ are the only atoms holding in $\struct{S}$. 
  Let $<$, $\rho$, and  $\mathcal{K}^{<}_{\mathrm{HN}}$ be as in~Theorem~\ref{thm:hubicka_nesetril}  for the set  $\cplmt{\mathcal{F}}\cup\mathcal G$.  
Define $\mathcal K$ as the subclass of all correctly labelled elements of $\mathcal{K}^{<}_{\mathrm{HN}}$. We remark that the predicates in $\rho$ are disjoint from those in $\sigma_+$, since they denote pieces of structures in the signature $\tau\cup\cplmt{\sigma}\cup\sigma_+$ (even if the structures of  $\cplmt{\mathcal{F}}$ have empty $\sigma_+$-relations).
The following auxiliary statement is a straightforward consequence of the two facts that $\mathcal{K}^{<}_{\mathrm{HN}}$ has the AP and the RP, and that the structures in $\mathcal{K}$ and $\cplmt{\mathcal{F}}$ are correctly labelled. 
 
\begin{claim}
  \label{lemma:AP}  
  $\mathcal{K}$ has the AP and the RP.
\end{claim}  

 Since $\mathcal{K}$ trivially has the \emph{joint embedding property} (the AP restricted to pairs $\struct{A},\struct{B}\in \mathcal{K}$ with $A \cap B = \emptyset$, cf.~\cite{hodges_book}), the AP is in fact implied by the RP~\cite{nesetril2005}. 
For illustrative purposes, we choose not to rely on this fact and instead first prove the AP as a special case of the RP. We also remark that the AP follows immediately from the fact that the class provided by Theorem~\ref{thm:hubicka_nesetril} has the AP, but we chose to present the proof of the AP as otherwise we  would have to prove the RP anyway.

\medskip 
\begin{claimproof} 
Let $\struct{C}$ and $\struct{D}$ be arbitrary elements of $\mathcal{K}$ whose substructures on $C\cap D$ are identical.
Clearly, $\struct{C}$ and $\struct{D}$ belong to  $\mathcal{K}^{<}_{\mathrm{HN}}$, and,
by Theorem~\ref{thm:hubicka_nesetril}, $\mathcal{K}^{<}_{\mathrm{HN}}$ has the AP.
Hence, there exists an amalgamation witness $\struct{W'}\in \mathcal{K}^{<}_{\mathrm{HN}}$ for $(\struct{C},\struct{D})$ such that $\struct C$ and $\struct D$ are substructures of $\struct{W'}$. 
We obtain $\struct W$ from $\struct W'$ by removing all tuples from relations with a symbol in $\tau\cup \overline{\sigma}$ witnessing that $\struct W'$ is incorrectly labelled.
More specifically, we repeat the following as many times as necessary: for every $R\in \tau \cup \overline{\sigma}$ and every $\bar{t}\in R^{\struct{W}'}$ such that there exists $X\in \sigma$ and $\bar{s}\subseteq \bar{t}$ matching the arity of $X$ while $\bar{s}\notin X^{\struct{W}'}\cup {\overline{X}}^{\struct{W}\smash{'}}$, we remove $\bar{t}$ from $R^{\struct{W}'}$.
We claim that \eqref{eq:HN} of Theorem~\ref{thm:hubicka_nesetril} still holds. Indeed, note first that all predicates in $\rho$ are pieces of structures in $\cplmt{\mathcal{F}}$, since the structures of $\mathcal{G}$ do not contain any pieces. 
Now if a $\rho$-predicate holds for a tuple in $\struct W'$, then the witness (according to the equivalence in \eqref{eq:HN} of Theorem~\ref{thm:hubicka_nesetril}) of this fact remains a witness in $\struct W$  because all witnesses in $\cplmt{\mathcal{F}}$ are correctly labelled, and therefore they remain unaltered.  Conversely, if a $\rho$-predicate does not hold for a tuple in $\struct W'$, then it also does not hold in $\struct W$, as the removal only reduces witnesses.  
It is now easy to see that $\struct W$ belongs to $\mathcal{K}$.
Also, since $\struct{C}$ and $\struct{D}$ are correctly labelled, no tuple that is contained entirely in $C$ or $D$ is ever removed, and hence  the modified structure $\struct W$  is still an amalgamation witness for $(\struct{C},\struct{D})$.

The RP can be shown in the same way as the AP, by removing tuples from $\tau$-relations  from a Ramsey witness $\struct W'$ in $\mathcal{K}^{<}_{\mathrm{HN}}$ for $(\struct{C},\struct{D})$ and $k\in \mathbb{N}$.
One important detail is that the removal of tuples might introduce new copies of $\struct{C}$.
But these copies cannot be contained in any old copy of $\struct{D}$ because it is correctly labelled.
\end{claimproof} 
 
\smallskip{\emph{Step 5: The final first-order expansion.}}
Let $\struct U$ be the Fra\"{i}ss\'{e}-limit of $\mathcal K$. Then  $\age(\struct U^{\tau\cup\cplmt{\sigma}})$ consists  precisely of the models of the first-order part of $\cplmt{\Phi}$; however,  the  expansion of $\struct U^{\tau\cup\cplmt{\sigma}}$ by $<$ might not be $\tau$-edge homogeneous. 
We will resolve this by adding redefined relations with symbols in $\sigma_+$, and then modifying $\cplmt{\Phi}$ accordingly. 

As for the former, consider every $(\struct P\cup\bar s,\bar s)$ in $\sigma_+$ obtained by extending a piece $(\struct P,\bar t)$ by the tuple   $\bar s$ and imposing   $R(\bar s)$ for some $R\in \tau$. On tuples in $R^{\struct{U}}$, the corresponding predicate currently never holds, as ensured by $\mathcal G$ above; on all other tuples, the predicate is random in the sense that no restrictions have been imposed in the construction of $\struct U$. 
We now add  all tuples in $R^{\struct{U}}$ to the relation interpreting $(\struct P\cup\bar s,\bar s)$  for which the implication $(\Leftarrow)$ in  \eqref{eq:HN} of Theorem~\ref{thm:hubicka_nesetril} requires it, thereby obtaining a structure $\struct{U_{\mo}}$. Note that the $\sigma_+$-relations of $\struct{U_{\mo}}$ are first-order definable from those of $\struct U$, and vice-versa; hence, the two structures have the same automorphism group.
Denote $ \cplmt{\sigma}_+\coloneqq {\cplmt{\sigma}\cup\sigma_+}$.
All $\rho$-relations being first-order definable in $\struct U^{\tau\cup \cplmt{\sigma}_+}$, we next conclude  that $\Aut(\struct U_{\mo}^{\tau\cup \cplmt{\sigma}_+},<)=\Aut(\struct{U})$. 
It follows immediately that   $(\struct U_{\mo}^{\tau\cup \cplmt{\sigma}_+},<)$ is Ramsey and $\omega$-categorical. 
It remains to prove that it is $\tau$-edge homogeneous; this can be done using the two facts that $\struct{U}$ is homogeneous and that the pieces of structures in $\cplmt{\mathcal{F}}$ are guarded in  $\sigma_+$.

\begin{claim}    \label{lemma:tau_edge_homogeneity}  
 $(\struct U_{\mo}^{\tau\cup \cplmt{\sigma}_+},<)$ is $\tau$-edge homogeneous. 
\end{claim} 
\begin{claimproof} Let $i$ be a partial isomorphism on $(\struct U_{\mo}^{\tau\cup\cplmt{\sigma}_+},<)$, defined on  a tuple $\bar u$ which is contained in some $\tau$-relation. By the homogeneity of $\struct U$, the map $i$ extends to an automorphism of $\struct U$, and hence of $(\struct U_{\mo}^{\tau\cup\cplmt{\sigma}_+},<)$ which has the same automorphisms, if $i$ is also a partial isomorphism on $\struct U$. 
The only potential obstacle to this are the predicates with a symbol in $\sigma_+$ (since these have been modified passing from $\struct U$ to $\struct U_{\mo}$) and in $\rho$ (since we are considering the  reduct without the $\rho$-relations). 

As for the former, suppose that for some tuple $\bar u'\subseteq\bar u$, we have that $(\struct P\cup \bar s,\bar s)(\bar u')$ holds in $\struct U$, where $(\struct P\cup \bar s,\bar s)$ is a  $\sigma_+$-symbol   extending a $\rho$-piece $(\struct P,\bar t)$ with $R(\bar s)$ for some $R\in \tau$. Since the construction of the original $\sigma_+$-relations involved the set $\mathcal G$, this implies that $R^{\struct U}(\bar u')$ does not hold, as atomic formulas with symbols $R$ and $(\struct P\cup \bar s,\bar s)$ never hold in $\struct{U}$ simultaneously. Hence, in the construction of $\struct U_{\mo}$, the relation interpreting $(\struct P\cup \bar s,\bar s)$ is neither modified on $\bar u'$ nor on its image under $i$, and since $i$ is a partial isomorphism on $\struct U_{\mo}^{\tau\cup\cplmt{\sigma}_+}$, we have that $(\struct P\cup \bar s,\bar s)(i(\bar u'))$ holds also in $\struct U$. Together with the same argument for  the inverse of $i$ we get that $i$ is a partial isomorphism on $\struct{U}$ even with respect to the original $\sigma_+$-relations.

Now consider the  case of $\rho$-relations. 

Suppose that, for some tuple $\bar u'\subseteq\bar u$ and for some $(\struct P,\bar t)(\bar x)\in\rho$, $(\struct P,\bar t)(\bar u')$ holds in $\struct{U}$.
As $\bar u\in R^\struct{U}$ for some $R\in\tau$, we have that $(\struct P\cup \bar s,\bar s)(\bar u)$ holds in $\struct{U}_{\mo}$, where $\bar s$ extends $\bar t$ in the same way as $\bar u$ extends $\bar u'$.
As $i$ is a partial isomorphism, we also have that $(\struct P\cup \bar s,\bar s)(i(\bar u))$ holds in $\struct{U}_{\mo}$, which implies that $(\struct P,\bar t)(i(\bar u'))$ holds in $\struct{U}$.
 
Again, with the same argument for the inverse of $i$, we conclude that $i$ is indeed a partial isomorphism on $\struct U$.
\end{claimproof}

\smallskip{\emph{Step 6: Back to GMSNP.}}
As announced in the previous step, we shall now expand $\cplmt{\Phi}$ by new clauses defining the predicates with a symbol in $\sigma_+$. 
Let $n$ be the size of the largest element of $\cplmt {\mathcal{F}}$. 
We add to the first-order part of $\cplmt{\Phi}$ the following.
\begin{enumerate}[label={\roman*.}]  
    \item \label{item:gmsnpclauses1} Each clause over 
    $\tau\cup\cplmt{\sigma}_+$ of the form $
    \psi(\bar{x},\bar{y}) \wedge R'(\bar x) \Rightarrow (\struct P'\cup \bar s',\bar s')(\bar{x})$ with at most $n$ variables, where $R'\in \tau$ and $R'(\bar s')$
    holds in the piece $(\struct P'\cup \bar s',\bar s') \in \sigma_+$ and $\psi$ is a conjunction of positive atoms with the following property: 
\begin{center} \vspace{0.75em}
\framebox{\parbox{0.74\textwidth}{ 
Replacing every subconjunction $(\struct P\cup \bar s,\bar s)(\bar{u})\wedge R(\bar u)$ in $\psi$ with $\pre{(\struct P\cup \bar s,\bar s)(\bar{u})}\wedge R(\bar u)$ yields a $(\tau\cup \cplmt{\sigma}_+)$-formula $\psi^{-1}$ such that there exists a homomorphism $h$ from the canonical database of \mbox{$\pre{(\struct P'\cup \bar s',\bar s')(\bar{x})}$ to the canonical database of $\psi^{-1}$ with $h(\bar{x})=\bar{x}$.} 
}} \vspace{0.75em}
\end{center} 
    \item \label{item:gmsnpclauses2}  Each clause over 
    $\tau\cup\cplmt{\sigma}_+$ with at most $n$ variables whose forbidden pattern is a conjunction of positive atoms with the following property:
\begin{center} \vspace{0.75em}
\framebox{\parbox{0.64\textwidth}{ 
Replacing every subconjunction $(\struct P\cup \bar s,\bar s)(\bar{u})\wedge R(\bar u)$ with $\pre{(\struct P\cup \bar s,\bar s)(\bar{u})}\wedge R(\bar u)$ yields a formula whose canonical database admits a homomorphism from a member of $\cplmt{\mathcal{F}}$. 
}} \vspace{0.75em}
\end{center}  
    \end{enumerate}
\begin{claim} \label{claim:recolouring}
    $\age(\struct U_{\mo}^{\tau\cup\cplmt{\sigma}_+})=\efm(\cplmt{\Phi}_+)$.
\end{claim}
\begin{claimproof}  
First, we show that every finite model of the first-order part of $\cplmt{\Phi}_+$ can be embedded into $\struct U_{\mo}^{\tau\cup\cplmt{\sigma}_+}$. 
To this end, let $\struct{A}\in \efm(\cplmt{\Phi}_+)$ be arbitrary.
Let $\phi$ be its canonical query, and modify $\phi$ by replacing conjuncts of the form $(\struct P\cup \bar s,\bar s)(\bar u)\wedge R(\bar u)$ with canonical conjunctive queries of the corresponding $(\tau\cup\cplmt{\sigma})$-structures $(\struct P\cup \bar s,\bar s)$.
We denote the resulting $(\tau\cup\cplmt{\sigma})$-formula by $\phi^{-1}$.
Let $\struct A'$ be the canonical database of $\phi^{-1}$, viewed as a $(\tau\cup\cplmt{\sigma}_+)$-structure. 
Note that, for every tuple $\bar u$ of elements of $\struct{A}'$, whenever $(\struct P\cup \bar s,\bar s)(\bar u)$ holds in $\struct{A}'$, then $R(\bar u)$ does not hold in $\struct{A}'$, where $(\struct P\cup \bar s,\bar s)\in\sigma_+$ is obtained from a piece $(\struct P,\bar t)$ by imposing $R$ on $\bar s\supseteq \bar t$. Hence, $\struct A'$ is contained in $\Forb_h(\mathcal G)$.
Moreover, $\struct A'$ is in $\Forb_{h}(\cplmt{\mathcal{F}})$ because $\struct{A}$ satisfies the clauses in item~\ref{item:gmsnpclauses2} for all possible variable substitutions. 
It is also correctly labelled since  $\struct A$ is (being a model of the first-order part of $\cplmt{\Phi}_+$) and since all elements of $\cplmt{\mathcal{F}}$ are (by definition). Expanding $\struct A'$ by relations in $\rho$ according to~\eqref{eq:HN} of Theorem~\ref{thm:hubicka_nesetril}  and by an arbitrary linear order $<$, we therefore obtain an embedding $i$ into  $\struct{U}$. 

We claim that $i$ restricted to the domain $A$ of $\struct A$  is an embedding of  $\struct A$ into ${\struct U^{\tau\cup\cplmt{\sigma}_+}_{\mo}}$, showing that  $\struct A$ indeed belongs to $\age(\struct U_{\mo}^{\tau\cup\cplmt{\sigma}_+})$. 
Since only relations with a symbol in $\sigma_+$ were modified in $\struct{U}_{\mo}$, we must only verify this for such relations.
Moreover, we must only consider particular $\tau$-guarded tuples in said relations.
 To this end, choose any $(\struct P\cup\bar s,\bar s)\in\sigma_+$ obtained by extending a piece $(\struct P,\bar t)$ by the tuple  $\bar s$ and imposing   $R(\bar s)$ for some $R\in \tau$. 
Let $\bar u$ be a tuple over $A$ such that $R(\bar u)$ holds in $\struct{A}$.

First, suppose that $(\struct P\cup\bar s,\bar s)(\bar u)$ holds in $\struct{A}$.
Then, by the construction of $\struct A'$, there exists a tuple $\bar u'$ over $A$ with $\bar u'\subseteq \bar u$  and such that  $\prexists{(\struct P,\bar t)(\bar{u}')}$ holds in $\struct A'$.
Since $i$ is a homomorphism, we have that $\prexists{(\struct P,\bar t)(i(\bar{u}'))}$ holds in $\struct U$. 
By the construction of $\struct{U}_{\mo}$, it follows that $(\struct P\cup\bar s,\bar s)(i(\bar u))$ holds in $\struct{U}_{\mo}$.

Second, suppose that $(\struct P\cup\bar s,\bar s)(i(\bar u))$ holds in $\struct{U}_{\mo}$.
By the construction of $\struct{U}_{\mo}$, it follows that $\prexists{(\struct P,\bar t)(i(\bar{u}'))}$ holds in $\struct U$ for a tuple $\bar u'$ over $A$ with $\bar u'\subseteq \bar u$.
Since $\struct U$ satisfies eq.~\eqref{eq:HN} of Theorem~\ref{thm:hubicka_nesetril}, also $(\struct P,\bar t)(i(\bar{u}'))$ holds in $\struct U$.
Since $i$ is an embedding with respect to  relations with a symbol in $\rho$, we have that $(\struct P,\bar t)(\bar u')$ holds in the $\rho$-expansion of $\struct{A}'$ witnessing the existence of $i$.
But since~\eqref{eq:HN} of Theorem~\ref{thm:hubicka_nesetril} also determines the relations with a symbol in $\rho$ in the said expansion of $\struct{A}'$, it must be the case that $\prexists{(\struct P,\bar t)(\bar{u}')}$ holds in $\struct A'$.
Since $\struct{A}$ satisfies the clauses in item~\ref{item:gmsnpclauses1} for all possible variable substitutions,  it follows that $(\struct P\cup\bar s,\bar s)(\bar u)$ holds in $\struct{A}$.

Next, we show that every element of $\age(\struct U_{\mo}^{\tau\cup\cplmt{\sigma}_+})$ satisfies the first-order part of $\cplmt{\Phi}_+$.
To this end, let $\struct{A}\in \age(\struct U_{\mo}^{\tau\cup\cplmt{\sigma}_+})$ be arbitrary; without loss of generality, $\struct{A}$ is a substructure of $\struct U_{\mo}^{\tau\cup\cplmt{\sigma}_+}$.
Then there exists a finite substructure $\struct{A}'$ of $\struct U_{\mo}^{\tau\cup\cplmt{\sigma}_+}$ extending $\struct{A}$ such that the existential witnesses for the primitive positive definitions of the $\cplmt{\sigma}_+$-relations of $\struct{A}$ in $\struct{U}$ are contained in $A'$.
The way how $\cplmt{\sigma}_+$-relations are defined in $\struct U_{\mo}^{\tau\cup\cplmt{\sigma}_+}$ from $\struct{U}$ over $A'$ and how $\struct{U}$ is constructed from $\cplmt{\mathcal{F}}\cup \mathcal{G}$ ensures that $\struct{A}$ satisfies the clauses in items~\ref{item:gmsnpclauses1} and~\ref{item:gmsnpclauses2} for all possible variable assignments.
Hence $\struct{A}\in \efm(\cplmt{\Phi}_+)$.
\end{claimproof}

Note that $\wh(\cplmt{\Phi}_+) = \wh(\Phi) = M$, $\ar(\cplmt{\Phi}_+)=(\Phi)$, and  $\hh(\cplmt{\Phi}_+) = |\sigma_+|$.
The size $\lh(\cplmt{\Phi}_+)$ is bounded from above by the number of $(\tau\cup\cplmt{\sigma}_+)$-structures of size at most $M+k$, so it is bounded by $2^{M^k|\sigma_+|}$. 
The statement of the theorem now follows by defining  $\Omega(\Phi)\coloneqq\cplmt{\Phi}_+$; we have established that $\fm(\cplmt{\Phi}_+)=\fm(\Phi)$. 
The property of recolouring-readiness of $\Omega(\Phi)$ is witnessed by setting $(\struct C_\Phi,<)\coloneqq(\struct U^{\tau\cup\cplmt{\sigma}_+},<)$.  
 \end{proof}

\subsection{Weak edge-recolourings} \label{sect:weak_edge}
Let $\struct{A}$ be a structure over a relational signature containing $\tau$.
We say that $\struct{A}$ is \emph{$\tau$-guarded} if there exist $R\in \tau$ and $\bar{t}\in R^{\struct{A}}$ such that $A\subseteq \bar{t}$. 
Now, let $\Phi_1$ and $\Phi_2$ be GMSNP $\tau$-sentences; as in Section~\ref{section:recolourings}, we set $n\coloneqq \max(\ar(\Phi_1),\ar(\Phi_2))$.  
A structure $\struct{A}$ with domain $[n]$ over a signature containing $<$ is \emph{standardly ordered} if $<^{\struct{A}}$ coincides with the natural ordering $1<\cdots < n$.

 A \emph{weak edge-recolouring} from $\Phi_1$ to $\Phi_2$ is a mapping $\xi_<$ from standardly ordered  $\tau$-guarded structures in $\colours(\Phi_1)^{<}$ to standardly ordered $\tau$-guarded structures in $\colours(\Phi_2)^{<}$ with the following two properties.
 First, the $\tau$-reduct of any structure and its $\xi_<$-image are isomorphic.
 Secondly, for every $\struct A\in  \efm(\Phi_1)^{<}$, there is a structure $\xi'_<(\struct A)\in \efm(\Phi_2)^{<}$ on the same domain such that ${<}^{\struct{A}}={<}^{\xi'_<(\struct A)}$ and, for all standardly ordered $\tau$-guarded $\struct{T} \in \colours(\Phi_1)^{<}$ and for every embedding $e\in \binom{\struct{A}}{\struct{T}}$, we have $e\in \binom{\xi'_<(\struct{A})}{\xi_<(\struct{T})}$.  
 In other words, the following  extension $\xi'_<$ is a well-defined mapping from $\efm(\Phi_1)^{<}$ to $\efm(\Phi_2)^{<}$:
\begin{center} \vspace{0.75em}
\framebox{\parbox{0.775\textwidth}{ 
 For every $\struct{A}\in \efm(\Phi_1)^{<}$, the structure $\xi'_<(\struct{A})$ on the same domain as $\struct{A}$ is obtained by removing all tuples that are not $\tau$-guarded from  all relations other than ${<}^{\struct A}$ and replacing every isomorphic copy of  a standardly ordered $\tau$-guarded  $\struct T \in \colours(\Phi_1)^{<}$  in $\struct{A}$ by an isomorphic  copy of $\xi_<(\struct T)$ whilst keeping the linear order ${<}^{\struct A}$ unchanged.
}} \vspace{0.75em}
\end{center}

Modulo the fact that the ordering of the colours matters in this definition, it is what one would intuitively understand under a recolouring of relational $\tau$-tuples.

The fact that the containment between weakly recolouring-ready GMSNP sentences can be reduced to the existence of such a recolouring can be proved almost exactly as in the proof of Proposition~\ref{prop:recolouring_containment}.

\begin{restatable}{proposition}{recolouringreadycontainment}   \label{prop:recolouring_ready_containment}
    For recolouring-ready GMSNP $\tau$-sentences $\Phi_1$ and $\Phi_2$, the following are equivalent:
    \begin{enumerate}
        \item \label{item:recolouring_containment1} $\fm(\Phi_1)\subseteq \fm(\Phi_2)$;
        \item \label{item:recolouring_containment2} There exists a weak edge-recolouring from $\Phi_1$ to $\Phi_2$. 
    \end{enumerate}   
\end{restatable}   
\begin{proof} Let $n\coloneqq \max(\ar(\Phi_1),\ar(\Phi_2))$.

``(\ref{item:recolouring_containment1})$\Rightarrow$(\ref{item:recolouring_containment2})''   
Let $(\struct{C}_{\Phi_1},<)$ and $(\struct{C}_{\Phi_2},<)$ be  two structures witnessing the recolouring-readiness of $\Phi_1$ and $\Phi_2$, respectively. 
Then, $\age(\struct{C}_{\Phi_1}
^{\tau}) \subseteq  \age(\struct{C}_{\Phi_2}
^{\tau})$.
By the properties (i) and (iii) of recolouring-readiness, we moreover have $\age(\struct{C}_{\Phi_1}^{\tau},<)\subseteq \age(\struct{C}_{\Phi_2}^{\tau},<)$.  
By the $\omega$-categoricity of $(\struct{C}_{\Phi_2}^{\tau},<)$,  Lemma~\ref{lemma:compactness} implies that there exists an embedding $e$ from $(\struct{C}_{\Phi_1}^{\tau},<)$ to $(\struct{C}_{\Phi_2}^{\tau},<)$.
By Theorem~\ref{th:canonical_ramsey}, we may assume that $e$ is canonical as a function from $(\struct{C}_1,<)$ to $(\struct{C}_2,<)$ since the former is an $\omega$-categorical Ramsey structure.
Let $\xi'_<$ be the mapping from $\efm(\Phi_1)^{<}$ to $\efm(\Phi_2)^{<}$ defined as follows.
Given $\struct{A}\in \efm(\Phi_1)^{<}$, consider an arbitrary embedding $i$ from $\struct{A}$ to $(\struct{C}_1,<)$.
Then, we define $\xi'_<(\struct{A}) \in \efm(\Phi_2)^{<}$ with domain $A$ whose relations are first defined through their preimages under $e\circ i$ and then pruned by removing all tuples that are not $\tau$-guarded from all relations other than $<$. 
Since $\Phi_2$ satisfies the monotonicity and the guarding axioms, the images of $\xi'_<$ are indeed in $\efm(\Phi_2)^{<}$.
Let $\xi_<$ be the restriction of $\xi'_<$ to standardly ordered $\tau$-guarded structures in $\colours(\Phi_1)^{<}$.
We claim that $\xi'_<$ arises from $\xi_<$ as in the definition of a recolouring. 

First, the $\tau$-reduct of any standardly ordered $\tau$-guarded  $\struct{T}\in\colours(\Phi_1)^{<}$  and its $\xi_<$-image are identical since $\xi_<(\struct{T})=\xi'_<(\struct{T})$ is obtained by pulling back relations via the $\tau$-embedding $e\circ i$.  
For the second part of the definition pertaining to $\xi'_<$, let $\struct A\in\efm{(\Phi_1)}^{<}$, standardly ordered $\tau$-guarded $\struct{T}\in\colours(\Phi_1)^{<}$,  and $f\in \binom{\struct A}{\struct{T}}$
be arbitrary; we claim that $f\in  \binom{\xi'_<(\struct A)}{\xi_<(\struct{T})} $.
By definition, the relations of $\xi'_<(\struct A)$ and of $\xi_<(\struct{T})$ are obtained by first embedding $\struct A$ and $\struct{T}$  into $(\struct{C}_1,<)$, then applying $e$, and then pulling back and pruning the relations from $(\struct{C}_2,<)$. 
This definition does not depend on the choice of the embeddings into $(\struct{C}_1,<)$, by the $\tau$-edge-homogeneity of $(\struct{C}_1,<)$ and the canonicity of $e$ -- the claim follows. 

``(\ref{item:recolouring_containment2})$\Rightarrow$(\ref{item:recolouring_containment1})''
Let $\struct{A}\in \fm(\Phi_1)$ be arbitrary; then there exists an expansion $\struct{A}'\in \efm(\Phi_1)^{<}$. By the definition of a recolouring we have $\xi'_<(\struct{A}')\in  \efm(\Phi_2)^{<}$. It also follows from this definition that  for all $\struct{T} \in \colours(\Phi_1)$ and every embedding $e\in \binom{\struct{A}}{\struct{T}}$ we have $e\in \binom{\xi'_<(\struct{A})}{\xi_<(\struct{T})}$. By our choice of $n$, this implies that the $\tau$-reduct of $\xi'_<(\struct{A}')$ coincides with $\struct A$. Hence $\xi'_<(\struct{A}')$ witnesses that  $\struct{A}\in \fm(\Phi_2)$.
\end{proof}

 \subsection{An order out of nowhere}\label{subsec:order_important}
 We now provide an example showing that  the dependence on an external linear order is a necessary component of the notion of a weak edge-recolouring for weakly recolouring-ready GMSNP sentences.
 To this end, consider the following two sentences $\Phi_1$ and $\Phi_2$ in GMSNP over the binary signature $\{E\}$.
 First, we obtain $\Phi_2$  from the sentence in eq.~\eqref{ex:gmsnp_introduction} by  forcing the colours $B$ (``blue'') and $R$ (``red'') to be synchronized with the orientation of the directed edges: 
 \begin{align*} 
  \exists R,B\, \forall x,y,z &\big(  \neg E(x,y) \vee   \neg  E(y,z) \vee \neg E(z,x) \vee \neg B(x,y) \vee   \neg  B(y,z) \vee \neg B(z,x)   
    \big)   \\
     {} \wedge \,  &\big(   \neg  E(x,y) \vee  \neg E(y,z) \vee \neg E(z,x) \vee  \neg R(x,y) \vee  \neg R(y,z) \vee \neg  R(z,x)   
    \big)   \\ 
   {} \wedge \,  & \big( \neg E(x,y) \vee \neg B(x,y) \vee \neg R(x,y) \big) \wedge  \big( \neg E(x,y) \vee B(x,y) \vee R(x,y) \big) \\
   {} \wedge \,  & \big( \neg E(x,y) \vee \neg B(y,x)   \big)  \wedge \big( \neg E(x,y) \vee \neg R(y,x)   \big)
\end{align*} 

 This in particular ensures that the structures in $\fm(\Phi_2)$  \emph{never} have symmetric $E$-edges.
 Second, we obtain $\Phi_1$  from $\Phi_2$ by allowing an additional binary colour $G$ (``green'') of directed edges while not imposing any additional restrictions except for again forcing colours to be synchronized with the orientation of the directed edges:
 \begin{align*}
  \exists R, G, B\, & \forall x,y,z \big(  \neg E(x,y) \vee   \neg  E(y,z) \vee \neg E(z,x) \vee \neg B(x,y) \vee   \neg  B(y,z) \vee \neg B(z,x)   
    \big)   \\
     {} \wedge \,  &\big(   \neg  E(x,y) \vee  \neg E(y,z) \vee \neg E(z,x) \vee  \neg R(x,y) \vee  \neg R(y,z) \vee \neg  R(z,x)   
    \big)   \\ 
   {} \wedge \,  & \big( \neg E(x,y) \vee    \neg B(x,y) \vee \neg R(x,y) \big) \wedge \big( \neg E(x,y) \vee   \neg G(x,y)  \vee \neg R(x,y) \big)  \\  
   {} \wedge \,  & \big( \neg E(x,y) \vee   \neg G(x,y) \vee \neg B(x,y) \big) \wedge \big( \neg E(x,y) \vee G(x,y) \vee B(x,y) \vee R(x,y) \big)    \\
   {} \wedge \, & \big( \neg E(x,y) \vee \neg G(y,x) \big) \wedge    \big( \neg E(x,y) \vee \neg B(y,x)   \big)  \wedge \big( \neg E(x,y) \vee \neg R(y,x)   \big)
\end{align*}

By a similar argument as the one given below Theorem~\ref{thm:hubicka_nesetril}, for both $i\in [2]$, the class $\efm(\Phi_i)$ has the free amalgamation property.
It then follows from the theorem of Ne\v{s}et\v{r}il and R\"{o}dl~\cite{NESETRIL1983183} that, for both $i\in [2]$, $\efm(\Phi_i)^{<}$ has the AP and the RP. 
For $i\in [2]$, let $(\struct{C}_{\Phi_i},<)$ be the Fra\"{i}ss\'{e}-limit of $\efm(\Phi_i)^{<}$. 
Then the homogeneous Ramsey structure $(\struct{C}_{\Phi_i},<)$ directly witnesses that $\Phi_i$ is weakly recolouring-ready.

Observe that $\fm(\Phi_1)$  consists of all finite directed graphs as we can simply colour every directed edge with ``green''.
Also $\fm(\Phi_2)$  consists of all finite directed graphs.
Indeed, for any directed graph $\struct{G}$, we can consider an arbitrary linear-order expansion $(\struct{G},<)$ thereof.
The directed edges which are oriented in the direction of $<$ will be coloured blue, and the remaining directed edges will be coloured red.
This colouring witnesses that $\struct{G}\in \fm(\Phi_2)$ because the added linear order is acyclic.
Hence, $\fm(\Phi_1) \subseteq \fm(\Phi_2)$.
By Proposition~\ref{prop:recolouring_ready_containment}, there must exist a weak edge-recolouring $\xi_<$ from $\Phi_1$ to $\Phi_2$.
In fact, we can easily provide a concrete weak edge-recolouring based on the above idea of taking linear-order expansions of directed graphs.
For every $C\in\{R,G,B\}$:

\smallskip 
\begin{itemize}
    \item $E(x,y) \wedge C(x,y)\wedge (x<y)$ is mapped to $E(x,y) \wedge R(x,y)\wedge (x<y)$;
    \item $E(x,y) \wedge C(x,y)\wedge (x>y)$ is mapped to $E(x,y) \wedge B(x,y)\wedge (x>y)$.
\end{itemize}  
\smallskip 

Finally, note that there is no weak edge-recolouring from $\Phi_1$ to $\Phi_2$ that would be independent of the linear ordering, i.e., by uniformly mapping $E(x,y) \wedge G(x,y)$ to either $E(x,y) \wedge R(x,y)$ or $E(x,y) \wedge B(x,y)$.
The reason is that red and blue directed cycles are forbidden in $\efm(\Phi_2)$ while green directed cycles are allowed in $\efm(\Phi_1)$.
In other words, weak edge-recolourings for weakly recolouring-ready GMSNP sentences are fundamentally linear order-dependent.
Similar phenomena are frequent in research on infinite-domain CSPs and their understanding is of critical importance for achieving further progress in the field~\cite{mottet2024order,pinsker2022current}. 
In the next section, we explain how to circumvent this instance of an order out of nowhere by encoding a ``local linear order'' into the second-order variables of $\Phi_1$ and $\Phi_2$; this will be an essential ingredient in the notion of recolouring-readiness.

\subsection{Recolouring-readiness} \label{sect:recolouring_readiness}
Let $\Phi$ be a connected GMSNP $\tau$-sentence, and let $n\coloneqq \ar(\Phi)$.
We say that $\Phi$ is \emph{edge-partitioned} if there exists a bijection $\pi$  from the $\tau$-guarded structures in $\colours(\Phi)$ to $\sigma$ such that for every $\struct{T}\in \colours(\Phi)$ and every $\struct{A} \in \efm(\Phi)$, we have 
$\struct{A} \models \pi(\struct{T})(a_1,\dots, a_n)$ if and only if $i \mapsto a_i$ is an embedding from $\struct{T}$ to $\struct{A}$. 
Fix $k\in \mathbb{N}$.
We say that $\Phi$ is $k$-\emph{acyclic} if there exists a function $\lambda$ from the $\tau$-guarded structures in $\colours(\Phi)$ to the $\tau$-guarded structures in $\colours(\Phi)^<$ such that the following two conditions are satisfied.
First, $\lambda(\struct{T})$ is a linear-order expansion of $\struct{T}$.
Secondly, for every $(\tau\cup \sigma)$-structure $\struct{A}\in \efm(\Phi)$ with $|A| \leq k$, there exists a linear-order expansion $\lambda'(\struct A)$ of $\struct A$ such that, for every $\tau$-guarded $\struct{T} \in \colours(\Phi)$ and  every embedding $e\in \smash{\binom{\struct{A}}{\struct{T}}}$, we have $e\in \smash{\binom{\lambda'(\struct{A})}{\lambda(\struct{T})}}$.
Finally,   we say that $\Phi$ is $k$-\emph{recolouring-ready} if it is weakly recolouring-ready, edge-partitioned, and $k$-acyclic.
Below, we restate Theorem~\ref{thm:recolouring_readiness} and provide a full proof thereof.
\recolouringreadiness*
\begin{figure}[ht]
     \centering
      \includegraphics[width=0.4\textwidth]{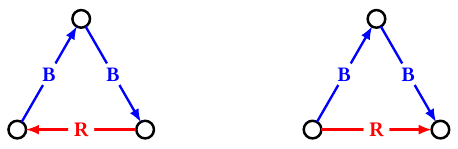}
     \caption{Two $\{E\}$-guarded structures in $\colours(\Phi_1)$ and $\colours(\Phi_2)$ for $\Phi_1$ and $\Phi_2$ as in Section~\ref{subsec:order_important} witnessing that these weakly recolouring-ready GMSNP sentences are not $3$-acyclic.}
     \label{fig:not_3_acyclic}
 \end{figure} 
Before we prove Theorem~\ref{thm:recolouring_readiness}, note that the sentences $\Phi_1$ and $\Phi_2$ in Section~\ref{subsec:order_important} are not $3$-acyclic.
Indeed, every choice of $\lambda$  induces a directed $<$-cycle in the $\lambda'$-image of the canonical database of one of the following two formulas (see Figure~\ref{fig:not_3_acyclic}): 
\begin{align*}
 &  E(x,y)\wedge E(y,z) \wedge E(z,x) \wedge B(x,y) \wedge B(y,z) \wedge R(y,z), \text{ or} \\
 &  E(x,y)\wedge E(y,z) \wedge E(z,x) \wedge B(x,y) \wedge B(y,z) \wedge R(z,y).  
\end{align*}    
\begin{proof}[Proof of Theorem~\ref{thm:recolouring_readiness}]
The proof of the first part extends the proof of Theorem~\ref{thm:recolouring_readiness2}.
For this reason, we only explain how to adjust the original proof.
The idea is to add, in Step~1, a fresh binary symbol $<$ to the signature $\overline{\sigma}$; we denote the so-obtained signature by $\overline{\sigma}_<$.
Then, we include in $\overline{\Phi}$ all clauses $\phi$ with at most $k$ variables so that the canonical database of the forbidden pattern of $\phi$ is a connected $\{<\}$-structure containing a directed $<$-cycle.
We also include the clauses 
\begin{align*} 
     \neg \big(  R(\bar{x}) \wedge  \neg (x_1<x_2) \wedge \neg  (x_2<x_1)\big) \quad \text{and} \quad  \neg \big(   R(\bar{x}) \wedge (x_1<x_2) \wedge  (x_2<x_1) \big)
\end{align*}  
for all symbols $R\in \tau\cup \overline{\sigma}$ and all $x_1,x_2\in \bar{x}$, where $\bar{x}$ is a tuple of fresh variables matching the arity of $R$.
This ensures that the so-obtained GMSNP sentence $\overline{\Phi}_{+,<}$ is $k$-acylic.
Finally, we modify $\overline{\Phi}_{+,<}$ as follows.
Let $n\coloneqq \ar(\overline{\Phi}_{+,<})$.
For every $\tau$-guarded $\struct{T}\in \colours(\overline{\Phi}_{+,<})$ of size $k\leq n$, we introduce a new $k$-ary existential symbol $X_{\struct{T}}$; the set of all such symbols is denoted by $\sigma_{+,<,p}$.
The \emph{atomic diagram} of a structure $\struct{T}$ is the conjunction  extending  its canonical query by listing not only all the atoms that hold in $\struct{T}$, but also their negations~\cite{hodges_book}.
We replace the clauses of $\overline{\Phi}_{+,<}$ by all possible clauses $\psi$ over $\tau\cup \sigma_{+,<,p}$ with at most $\big(\wh(\overline{\Phi}_{+,<})+\ar(\overline{\Phi}_{+,<})\big)$-many variables satisfying the monotonicity and the guarding axioms such that replacing every atom $X_{\struct{T}}(\bar{x})$ in the forbidden pattern of $\psi$ by the atomic diagram of $\struct{T}$ yields a formula that is not satisfiable together with the first-order part of $\overline{\Phi}_{+,<}$.  
This ensures that the so-obtained GMSNP $\tau$-sentence $\overline{\Phi}_{+,<,p}$ is edge-partitioned.
Note that a similar trick was used in the proof of Proposition~\ref{prop:from_GMSNP_to_SNP_with_AP_and_RP}.
It is not hard to see that $\overline{\Phi}_{+,<,p}$ is $k$-acyclic because $\overline{\Phi}_{+,<}$ is $k$-acyclic. 
Moreover, we clearly have $\fm(\overline{\Phi}_{+,<,p})=\fm(\overline{\Phi}_{+,<})$.
It remains to find a structure witnessing that $\overline{\Phi}_{+,<,p}$ is weakly recolouring-ready.
Such a structure can be first-order defined from $\struct{U}$ similarly as $\struct U_{\mo}^{\tau\cup\cplmt{\sigma}_+}$.

Now, we prove the second part of the theorem.
Let $k\in \mathbb{N}$ be such that $k\geq \max(\wh(\Phi_1),\wh(\Phi_2))$, and let $\Phi_1$ and $\Phi_2$ be $k$-recolouring-ready GMSNP $\tau$-sentences.

First, suppose that $\fm(\Phi_1)\subseteq \fm(\Phi_2)$.
Since $\Phi_1$ and $\Phi_2$ are $k$-recolouring-ready, they are in particular weakly recolouring-ready.
Hence, by Proposition~\ref{prop:recolouring_ready_containment}, there exists a weak edge-recolouring $\xi_{<}$ from $\Phi_1$ to $\Phi_2$.
We claim that there exists an edge-recolouring $\xi$ from $\Phi_1$ to $\Phi_2$. Recall that an edge-recolouring $\xi$ is a mapping as in eq.~\eqref{eq:edge_recolouring}.
For $i\in [2]$, let $\lambda_i$ and $\pi_i$ be functions witnessing $k$-acyclicity and edge-partitioning for $\Phi_i$, respectively.
Consider a conjunction  $R(x_1,\dots, x_n) \wedge \alpha_1(x_1,\dots, x_n)$,  
where $R\in \tau$ and $\alpha_1$ is an atomic $\sigma_1$-formula.
Then $\alpha_1$ is of the form $\pi_1(\struct{T}_1)$ for a $\tau$-guarded structure $\struct{T}_1\in \colours(\Phi_1)$ such that, for every $\struct{A} \in \efm(\Phi_1)$, we have $\struct{A} \models \alpha_1(a_1,\dots, a_n)$ if and only if $i \mapsto a_i$ is an embedding from $\struct{T}_1$ to $\struct{A}$. 
Recall that $\lambda_1(\struct{T}_1)$ is a linear-order expansion of $\struct{T}_1$.
We assume that $\lambda_1(\struct{T}_1)$ is standardly ordered, otherwise we reorder the variables $x_1,\dots, x_n$ accordingly; this assumption is of crucial importance for the rest of the proof.
By definition, $\xi_{<}$ maps the standardly ordered structure $(\struct{T}_1,<)$ to a standardly ordered structure
$(\struct{T}_2,<)$, where $\struct{T}_2\in \colours(\Phi_2)$.
We define the image of $R(x_1,\dots, x_n) \wedge \alpha_1(x_1,\dots, x_n)$ under $\xi$ as 
$\pi_2(\struct{T}_2)(x_1,\dots, x_n) \wedge \alpha_2(x_1,\dots, x_n)$.
It remains to verify that $\xi$ is an edge-recolouring from $\Phi_1$ to $\Phi_2$.
Let $\struct{A}_1\in \efm(\Phi_1)$ be arbitrary, and let $\struct{A}_2\coloneqq \xi'(\struct{A}_1)$ be the $(\tau\cup \sigma_2)$-structure obtained from $\struct{A}_1$ by applying $\xi$ to $\tau$-guarded tuples in $\sigma_1$-relations and removing non-$\tau$-guarded tuples from $\sigma_1$-relations.
Without loss of generality, we may assume that the domain size of $\struct{A}_1$ is at most $\wh(\Phi_2)\leq k$.
By the definition of $\lambda_1$, there exists a linear-order expansion $\lambda_1'(\struct{A}_1)$ of $\struct{A}_1$ such that, for every $\tau$-guarded $\struct{T}_1 \in \colours(\Phi_1)$ and  every embedding $e\in  \binom{\struct{A}_1}{\struct{T}_1}$, we have $e\in \binom{\lambda_1'(\struct{A}_1)}{\lambda_1(\struct{T}_1)}$.
Since $\lambda_1'(\struct{A}_1)\in \efm(\Phi_1)^{<}$ and $\xi_{<}$ is a weak edge-recolouring, the extension $\xi'_{<}$ is a well-defined mapping from $\efm(\Phi_1)^{<}$ to $\efm(\Phi_2)^{<}$.
Finally, since $\xi$ is defined via $\xi_{<}$, we have that $\struct{A}_2 = \xi'(\struct{A}_1)$ satisfies the first-order part of $\Phi_2$.  

 Conversely, similarly as in Proposition~\ref{prop:recolouring_ready_containment}, given an edge-recolouring from $\Phi_1$ to $\Phi_2$, we trivially get that $\fm(\Phi_1)\subseteq \fm(\Phi_2)$. 
\end{proof}

\section{Conclusion and outlook}

We proved the decidability of the containment and the FO-rewritability problems for GMSNP, thereby settling an open question posed in~\cite{bienvenu2014,bouhris_lutz2016}.
Our decision procedure runs in non-deterministic 2-exponential time, which exactly matches the lower bound on the complexity obtained in~\cite{bouhris_lutz2016}.

In the proof of Theorem~\ref{thm:2NEXPTIME_for_GMSNP}, we employ structural Ramsey theory to effectively reduce from the containment problem for GMSNP to the problem of testing the existence of a recolouring between SNP sentences.
As mentioned in the introduction, the use of structural Ramsey theory is only one of the possible ways to outsource combinatorics.
In fact, in all works on the containment problem for MMSNP except~\cite{bodirsky2018_article}, a different tool was used, 
commonly known as the \emph{sparse incomparability lemma}~\cite[Thm.~1]{kun2013} (Feder and Vardi originally used a weaker, randomised version of this result, which they attributed to Erd\H{o}s~\cite[Thm.~5]{federvardi1998}). 
It would be interesting to know whether the original approach can be used to reprove our 
\TWONEXPTIME-upper bound on the complexity of the containment problem for GMSNP.

Regarding our approach using structural Ramsey theory, one can clearly see that it has the potential to work in a broader setting than GMSNP.
More specifically, by Proposition~\ref{prop:recolouring_containment}, the $\TWONEXPTIME$ upper bound on the complexity of containment holds for all pairs of SNP sentences whose first-order part defines a class of finite structures with the AP and the RP.
The issue with this statement is that the said fragment of SNP does not have any known explicit description;
it might be difficult to judge Proposition~\ref{prop:recolouring_containment} as a stand-alone contribution 
since the complexity of recognizing which SNP sentences fall within its scope might be high, potentially even undecidable.
For a discussion of the complexity of related meta problems, see~\cite{rydval:LIPIcs.ICALP.2024.150,rydval_arxiv}.
We would be interested in locating virtually any fragment of existential second-order logic that properly generalises GMSNP, can be embedded into SNP with AP and RP similarly as GMSNP, and whose syntax is efficiently verifiable.

As our secondary contribution, we refined the construction of Bodirsky, Kn\"{a}uer, and Starke by adding a restricted form of homogeneity to the properties of these structures; the result is captured by the notion of recolouring-readiness.
In Theorem~\ref{thm:recolouring_readiness}, we showed that from every connected GMSNP sentence one can compute a logically equivalent connected recolouring-ready GMSNP sentence over the same signature.
Moreover the (computable) transformation to connected recolouring-ready GMSNP translates the containment between the finite models of connected recolouring-ready GMSNP sentences to the existence of an edge-recolouring between connected recolouring-ready GMSNP sentences.

\bibliographystyle{plain}
\bibliography{references}
\end{document}